\documentclass{article}
\pdfoutput=1

\usepackage[a4paper]{geometry}
\usepackage[T1]{fontenc}
\usepackage[utf8]{inputenc}
\usepackage{lmodern}
\usepackage{tikz}
\usetikzlibrary{calc,patterns,arrows.meta}
\usepackage{amsmath,amssymb}
\usepackage{amsthm}
\usepackage{subcaption}
\usepackage{booktabs}
\usepackage{xcolor}

\newtheorem{theorem}{Theorem}

\newtheorem*{theorem*}{Theorem}

\theoremstyle{remark}

\definecolor{col1}{RGB}{100,143,255}
\definecolor{col2}{RGB}{120, 94, 240}
\definecolor{col3}{RGB}{254,97,0}
\definecolor{col4}{RGB}{220, 38, 127}
\definecolor{col5}{RGB}{255, 176, 0}

\usepackage{hyperref}
\usepackage[backend=biber, sorting=none, maxnames=99, minnames=99, giveninits=true]{biblatex}

\newcommand{\dd}{\text{d}}
\newcommand{\tr}{\mathrm{tr}}
\newcommand{\ren}{\mathrm{ren}}

\newcommand{\Aut}{\mathrm{Aut}}

\newcommand{\R}{\mathbb{R}}
\newcommand{\Q}{\mathbb{Q}}

\newcommand{\Z}{\mathbb{Z}}

\newcommand{\HB}{\mathcal{H}}

\newcommand{\vertextwo}{%
\tikz[baseline=-0.6ex,scale=0.62,line cap=round,line join=round]{
  \coordinate (v) at (0,0);
  \foreach \ang in {175,195} {
    \draw (v) -- ++(\ang:0.95);
  }

  \filldraw[draw=black,thick] (v) circle (2pt);
}}

\newcommand{\vertexfour}{%
\tikz[baseline=-0.6ex,scale=0.62,line cap=round,line join=round]{
  \coordinate (v) at (0,0);
  \foreach \ang in {135,155,205,225} {
    \draw (v) -- ++(\ang:0.95);
  }

  \filldraw[draw=black,thick] (v) circle (2pt);
}}

\newcommand{\vertextwoloop}{%
\tikz[baseline=-0.6ex,scale=0.62,line cap=round,line join=round]{
  \coordinate (v) at (0,0);

  \foreach \ang in {175,195} {
    \draw (v) -- ++(\ang:0.95);
  }

  \draw (v)
    .. controls (0.45,0.38) and (1.05,0.32) .. (1.05,0)
    .. controls (1.05,-0.32) and (0.45,-0.38) .. (v);

  \filldraw[draw=black,thick] (v) circle (2pt);
}}

\newcommand{\vertextwolooptwolegs}{%
\tikz[baseline=-0.6ex,scale=0.62,line cap=round,line join=round]{
  \coordinate (v) at (0,0);

  \foreach \ang in {135,155,205,225} {
    \draw (v) -- ++(\ang:0.95);
  }

  \draw (v)
    .. controls (0.45,0.38) and (1.05,0.32) .. (1.05,0)
    .. controls (1.05,-0.32) and (0.45,-0.38) .. (v);

  \filldraw[draw=black,thick] (v) circle (2pt);
}}

\newcommand{\vertexloop}{%
\tikz[baseline=-0.6ex,scale=0.62,line cap=round,line join=round]{
  \coordinate (v) at (0,0);

  \draw (v)
    .. controls (0.45,0.38) and (1.05,0.32) .. (1.05,0)
    .. controls (1.05,-0.32) and (0.45,-0.38) .. (v);

  \filldraw[draw=black,thick] (v) circle (2pt);
}}

\newcommand{\vertextwolooptwo}{%
\tikz[baseline=-0.6ex,scale=0.62,line cap=round,line join=round]{
  \coordinate (v) at (0,0);

  \foreach \ang in {175,195} {
    \draw (v) -- ++(\ang:0.95);
  }

  \draw (v)
    .. controls (0.40,0.58) and (1.10,0.60) .. (1.18,0.25)
    .. controls (1.05,-0.03) and (0.42,0.16) .. (v);

  \draw (v)
    .. controls (0.42,-0.16) and (1.05,0.03) .. (1.18,-0.25)
    .. controls (1.10,-0.60) and (0.40,-0.58) .. (v);

  \filldraw[draw=black,thick] (v) circle (2pt);
}}

\newcommand{\vertextwolooptwotwolegs}{%
\tikz[baseline=-0.6ex,scale=0.62,line cap=round,line join=round]{
  \coordinate (v) at (0,0);

  \foreach \ang in {135,155,205,225} {
    \draw (v) -- ++(\ang:0.95);
  }

  \draw (v)
    .. controls (0.40,0.58) and (1.10,0.60) .. (1.18,0.25)
    .. controls (1.05,-0.03) and (0.42,0.16) .. (v);

  \draw (v)
    .. controls (0.42,-0.16) and (1.05,0.03) .. (1.18,-0.25)
    .. controls (1.10,-0.60) and (0.40,-0.58) .. (v);

  \filldraw[draw=black,thick] (v) circle (2pt);
}}

\newcommand{\bubbletwolegs}{%
\tikz[baseline=-0.6ex,scale=0.62,line cap=round,line join=round]{
  \coordinate (l) at (0,0);
  \coordinate (r) at (1.15,0);

  \foreach \ang in {135,225} {
    \draw (l) -- ++(\ang:0.95);
  }

  \draw (l) to[bend left=55] (r);
  \draw (l) to[bend right=55] (r);

  \filldraw[draw=black,thick] (l) circle (2pt);
  \filldraw[draw=black,thick] (r) circle (2pt);
}}

\newcommand{\bubblefourlegs}{%
\tikz[baseline=-0.6ex,scale=0.62,line cap=round,line join=round]{
  \coordinate (l) at (0,0);
  \coordinate (r) at (1.15,0);

  \foreach \ang in {135,165,195,225} {
    \draw (l) -- ++(\ang:0.95);
  }

  \draw (l) to[bend left=55] (r);
  \draw (l) to[bend right=55] (r);

  \filldraw[draw=black,thick] (l) circle (2pt);
  \filldraw[draw=black,thick] (r) circle (2pt);
}}

\newcommand{\bubbletwotwolegs}{%
\tikz[baseline=-0.6ex,scale=0.62,line cap=round,line join=round]{
  \coordinate (l) at (0,0);
  \coordinate (r) at (1.15,0);

  \foreach \ang in {135,225} {
    \draw (l) -- ++(\ang:0.95);
  }
  \foreach \ang in {45,315} {
    \draw (r) -- ++(\ang:0.95);
  }

  \draw (l) to[bend left=55] (r);
  \draw (l) to[bend right=55] (r);

  \filldraw[draw=black,thick] (l) circle (2pt);
  \filldraw[draw=black,thick] (r) circle (2pt);
}}

\newcommand{\sunrisetwolegs}{%
\tikz[baseline=-0.6ex,scale=0.62,line cap=round,line join=round]{
  \coordinate (l) at (0,0);
  \coordinate (r) at (1.15,0);

  \draw (l) -- ++(158:0.95);
  \draw (r) -- ++(22:0.95);

  \draw (l) to[bend left=55] (r);
  \draw (l) -- (r);
  \draw (l) to[bend right=55] (r);

  \filldraw[draw=black,thick] (l) circle (2pt);
  \filldraw[draw=black,thick] (r) circle (2pt);
}}

\newcommand{\bubbletadpolelegsright}{%
\tikz[baseline=-0.6ex,scale=0.62,line cap=round,line join=round]{
  \coordinate (l) at (0,0);
  \coordinate (r) at (1.15,0);

  \draw (l)
    .. controls (-0.45,0.38) and (-1.05,0.32) .. (-1.05,0)
    .. controls (-1.05,-0.32) and (-0.45,-0.38) .. (l);

  \draw (l) to[bend left=55] (r);
  \draw (l) to[bend right=55] (r);

  \foreach \ang in {45,315} {
    \draw (r) -- ++(\ang:0.95);
  }

  \filldraw[draw=black,thick] (l) circle (2pt);
  \filldraw[draw=black,thick] (r) circle (2pt);
}}

\newcommand{\bubbletadpolelegsleft}{%
\tikz[baseline=-0.6ex,scale=0.62,line cap=round,line join=round]{
  \coordinate (l) at (0,0);
  \coordinate (r) at (1.15,0);

  \draw (l)
    .. controls (-0.38,0.45) and (-0.32,1.02) .. (0,1.02)
    .. controls (0.32,1.02) and (0.38,0.45) .. (l);

  \draw (l) to[bend left=55] (r);
  \draw (l) to[bend right=55] (r);

  \foreach \ang in {160,200} {
    \draw (l) -- ++(\ang:0.95);
  }

  \filldraw[draw=black,thick] (l) circle (2pt);
  \filldraw[draw=black,thick] (r) circle (2pt);
}}

\newcommand{\vertexthreelooptwo}{%
\tikz[baseline=-0.6ex,scale=0.62,line cap=round,line join=round]{
  \coordinate (v) at (0,0);

  \foreach \ang in {175,195} {
    \draw (v) -- ++(\ang:0.95);
  }

  \draw (v)
    .. controls (0.30,0.62) and (1.02,0.80) .. (1.20,0.44)
    .. controls (1.14,0.12) and (0.40,0.26) .. (v);

  \draw (v)
    .. controls (0.45,0.30) and (1.22,0.26) .. (1.28,0)
    .. controls (1.22,-0.26) and (0.45,-0.30) .. (v);

  \draw (v)
    .. controls (0.40,-0.26) and (1.14,-0.12) .. (1.20,-0.44)
    .. controls (1.02,-0.80) and (0.30,-0.62) .. (v);

  \filldraw[draw=black,thick] (v) circle (2pt);
}}

\newcommand{\vertextworen}{%
\tikz[baseline=-0.6ex,scale=0.62,line cap=round,line join=round]{
  \coordinate (v) at (0,0);
  \foreach \ang in {175,195} {
    \draw (v) -- ++(\ang:0.95);
  }

  \filldraw[fill=white,draw=black,thick] (v) circle (0.11);
}}

\newcommand{\vertexfourren}{%
\tikz[baseline=-0.6ex,scale=0.62,line cap=round,line join=round]{
  \coordinate (v) at (0,0);
  \foreach \ang in {135,155,205,225} {
    \draw (v) -- ++(\ang:0.95);
  }

  \filldraw[fill=white,draw=black,thick] (v) circle (0.11);
}}

\hypersetup{
  pdfauthor={Michael Borinsky, Jury Radkovski, Suriyah Rajalingam Kannagi}, 
  pdftitle={Tropical QFT and the sine-Gordon model},
  pdfsubject={},
  pdfkeywords={},
  linkcolor  = col2,
  citecolor  = col4,
  urlcolor   = col5,
  colorlinks = true,
}

\title{Tropical QFT and the sine-Gordon model}
\author{Michael Borinsky$^{a}$\footnote{mborinsky@perimeterinstitute.ca} \and Jury Radkovski$^{a,b,c}$\footnote{j.radkovski@nbi.ku.dk} \and Suriyah Rajalingam Kannagi$^{a,d}$\footnote{srajalingamkannagi@perimeterinstitute.ca}~\\[2mm]
{\small\it $^a$Perimeter Institute for Theoretical Physics, 31 Caroline St N,  Waterloo, ON N2L 2Y5, Canada}\\
{\small \it $^b$Department of Physics and Astronomy, McMaster University,}\\
{\small \it 1280 Main St W, Hamilton, ON L8S 4M1, Canada}\\ 
{\small\it $^c$Center of Gravity, Niels Bohr Institute, Blegdamsvej 17, 2100 Copenhagen, Denmark}\\ 
{\small\it $^d$Department of Physics and Astronomy, University of Waterloo, Waterloo, ON N2L 3G1, Canada}
}
\date{}

\begin{document}

\maketitle
\begin{abstract}
A tropical quantum field theory is an exactly solvable deformation limit of an ordinary QFT, governed by a non-linear partial differential equation.  We use the integrable two-dimensional sine-Gordon theory as a testbench to compare the exact features of an undeformed model with its tropical counterpart.  We classify scale-invariant tropical QFTs in $D=2$ and leverage functional renormalization group methods to demonstrate that all scale-invariant theories with periodic effective actions are infrared fixed points of flows originating from sine-Gordon Lagrangians.  For the tropical sine-Gordon model, we compute the $\beta$-functions exactly, solve the renormalized theory both via a convergent expansion and numerically for all admissible couplings, and illustrate the complete renormalization group flow from the Gaussian to the IR fixed point.  Finally, we study the large-order behaviour of the loop expansion and relate it to the Hamilton--Jacobi equation, whose solutions we interpret as tropical instantons.
\end{abstract} 

\section{Introduction}
Quantum field theory (QFT) is highly successful in describing nature, but it is difficult to study. Interacting QFTs can only be solved in some very special cases, and the importance of discovering such special solutions was quickly realized after QFT was invented.

The first author recently introduced a class of theories that is solvable in a strong sense: \emph{tropical quantum field theories}~\cite{Borinsky:2025stz}. 
Tropical QFTs are deformation limits of ordinary QFTs. For instance, we deform a given scalar QFT with a real parameter $\xi > 0$ by replacing the spacetime dimension $D$ with $D \cdot \xi$ and the propagator $(\Box + m^2)^{-1}$ with $(\Box + m^2)^{-\xi}$. For $\xi = 1$, we recover the original QFT. In the \emph{tropical limit} $\xi \rightarrow 0^+$, the theory simplifies drastically, but it does not become trivial. The effective action becomes a function of the field and the coupling constants, and this function is fixed completely by a \emph{non-linear} partial differential equation, the \emph{tropical loop equation}, and appropriate boundary conditions, which in turn are fixed by the classical action. Tropical QFT is thus almost as easy (or hard) as quantum mechanics, which is governed by a \emph{linear} partial differential equation:  the Schr\"odinger equation.

Since tropical QFT is a deformation limit of QFT, we can compare it to the large-$N$ limit of colored theories. Viewed from a perturbative angle, the large-$N$ limit removes the non-planar Feynman graphs. In contrast, the tropical limit keeps all graphs, but replaces each Feynman integral by its tropical counterpart: the \emph{Hepp bound}~\cite{Panzer:2019yxl}, a rational function in $D$ with a pole wherever the integral diverges logarithmically. We can think of tropical QFT as a shadow of the original QFT: it keeps the singularity structure of all Feynman integrals in the spacetime dimension and (mostly) forgets their dependence on masses and momenta, provided that this dependence does not change the singularity structure of the integral.

One interesting application of this shadow to the original QFT is \emph{global tropical sampling}~\cite{Borinsky:2025stz}: the tropical solution gives a probability distribution on the Feynman graphs of a given loop order and their Schwinger parameters, and sampling from it allows one to estimate the full perturbative coefficient of the original QFT, without enumerating Feynman graphs. Remarkably, the algorithm needs only polynomial time and memory in the loop order and the number of legs. The evaluation of a single Feynman integral even in the scalar case, in Euclidean spacetime and using advanced methods takes exponential time~\cite{Borinsky:2020rqs}. The perturbative coefficients of massive scalar QFT in Euclidean spacetime thus become accessible at loop orders that, before, were out of reach by multiple orders of magnitude. As a proof of concept, the first author computed the primitive beta function of $\phi^4$ theory in $D=4$ at 50 loops~\cite{Borinsky:2025stz}.

Balduf and Panzer~\cite{Balduf:2025idm} used these tools to analyze renormalization and large-order behaviour of tropical $\phi^4$ theory in $D=4$. They found, for instance, that the renormalized tropical theory has no renormalon singularities in the minimal subtraction scheme, which suggests that the original $\phi^4$ theory shares this property in the same scheme. Their work also shows a second use of tropical QFT, central to this article: it is a testbench for QFT computations. The usual perturbative technology carries over, but the computations are generally simpler and usually hidden, non-perturbative, information becomes accessible.

Tropical geometry has already become an established tool for scattering amplitudes, Feynman integrals and quantum field theory in general. The field theory limit of string amplitudes can be understood as a tropical limit of the moduli space of curves \cite{wittenTropical,Tourkine:2013rda}.  Tropical and convex geometry play a central role in studying the amplituhedron \cite{Arkani-Hamed:2013jha,Arkani-Hamed:2017tmz,Arkani-Hamed:2019mrd} and scattering amplitudes in terms of moduli spaces of genus 0 tropical curves, also called moduli spaces of (phylogenetic) trees or tropical Grassmannians \cite{BILLERA2001733,MR2071813,Cachazo:2019ngv}. For individual Feynman integrals, tropical geometry facilitates the construction of local subtraction schemes \cite{Arkani-Hamed:2022cqe,Salvatori:2024nva,delaCruz:2024xsm}, most recently in the symbolic integration package SubTropica~\cite{Giroux:2026tgd}. Further, Mikhalkin pioneered the tropical approach to enumerative Gromov-Witten theory \cite{MR809718,Witten:1988xj,Kontsevich:1994qz,MR2137980}. Recently, this approach has been interpreted physically as the tropicalization of a sigma model~\cite{Albrychiewicz2024Tropological,Albrychiewicz2025TropicalBranes}.  The Hepp bound \cite{Panzer:2019yxl} and tropical sampling for Feynman integrals \cite{Borinsky:2020rqs,Borinsky:2023jdv,Borinsky:2025asc} are the direct precursors of tropical QFT. Important intermediate steps were large-loop order investigations of perturbation theory in scalar QFTs \cite{Balduf:2023ilc,Balduf:2024gvv,Balduf:2024njk,Borinsky:2025ncs}. Closest to our approach is the surfaceology program \cite{Arkani-Hamed:2023lbd,Arkani-Hamed:2024vna,Arkani-Hamed:2024pzc,De:2024wsy,Salvatori:2025oib}. There, amplitudes of colored theories are written as integrals over curves on a surface with integrands built from tropical functions. 

Section~\ref{sec:qfttr} of this article is a self-contained introduction to tropical QFT from the path integral point of view, with explicit example computations. The tropical loop equation holds beyond perturbation theory, so that we can ask non-perturbative questions, for instance about renormalization group fixed points. Here, a close connection to functional renormalization group technology \cite{Berges:2000ew,Nagy:2012ef,Dupuis:2020fhh} becomes useful. The tropical loop equation has the same form as the flow equation for the effective potential in the local potential approximation of the Wetterich equation \cite{Wetterich:1992yh,Litim:2001up}, and the tropical loop equation for the free energy has the form of the Polchinski equation~\cite{Polchinski:1983gv}. There, these equations are approximations, whereas in tropical QFT, they are exact. In Section~\ref{sec:trop_ren_group_fix}, we translate previous functional renormalization group results~\cite{Felder1987,Morris:1994jc,Nagy:2009pj,Codello2012} to statements about scale-invariant tropical QFTs. In two dimensions, the tropical loop equation for scale-invariant theories becomes an autonomous ordinary differential equation, which can be read as Newton's equation for a particle in a potential. Its real, periodic solutions form, up to a phase shift, a one-parameter family of oscillating effective actions, i.e.\ a continuum of scale-invariant tropical QFTs. 
Interestingly, Schaaf's theorem~\cite{MR814016}, which gives a criterion for the period of an oscillatory classical mechanical system to increase monotonically with energy, is important in clarifying the nature of scale-invariant tropical QFTs in $D=2$ with bounded effective action. 
For small oscillation amplitude, the effective action is a cosine, i.e.\ a sine-Gordon potential. 

For this reason, the second part of this paper starting from Section~\ref{sec:sinegordon} focuses on the $D=2$ tropical sine-Gordon model. Undeformed sine-Gordon theory is a well-studied two-dimensional model \cite{Coleman1975}. It is rigorously known that $D=2$ sine-Gordon theory exists~\cite{Frohlich:1975rcn}. In Section~\ref{sec:irfixedpoint}, we will show that general $D=2$ tropical QFTs and tropical sine-Gordon theory are closely related: every scale-invariant tropical QFT in $D=2$ with a real, periodic effective action is the IR fixed point of a renormalization group flow that starts from a bare sine-Gordon~Lagrangian. So, sine-Gordon Lagrangians classify the periodic scale-invariant tropical QFTs in $D=2$.

The tropical sine-Gordon model has the same UV divergences as the original model, and we remove them in the same way (Section~\ref{sec:renormalization}), paralleling Coleman's treatment~\cite{Coleman1975}. We compute the $\beta$-functions of the tropical model exactly. After the rescaling $g^2 \rightarrow g^2/(4\pi)$, they agree with the scheme-independent coefficients of the original model's $\beta$-functions~\cite{AmitGoldschmidtGrinstein1980}. In particular, the Coleman point $g^2 = 8\pi$, beyond which the original model has no ground state, corresponds to $g^2 = 2$ in the tropical model, which also has no solution beyond this point.

We solve the renormalized tropical model in several ways. The expansion in the renormalized mass-coupling $\alpha_\ren$ converges for small coupling if the other coupling constant is in the admissible range $0 < g^2 < 2$. We prove the convergence of this expansion rigorously in Theorem~\ref{thm:convergence} in Section~\ref{sec:solvingpde}. With this expansion as initial data, we integrate the renormalized tropical loop equation numerically up to arbitrarily large $\alpha_\ren$. This gives the effective action for all coupling values, and we can visually observe how the theory flows from the Gaussian fixed point in the UV to its scale-invariant IR fixed point, which corresponds to the limit $\alpha_\ren \rightarrow \infty$.

Exact information about the sine-Gordon theory is available: Lukyanov and Zamolodchikov~\cite{Lukyanov1996} conjectured a closed formula for the expectation values of powers of the field at the origin. In Section~\ref{subsec:comparison-sine-gordon}, we compare the tropical model with this formula and find qualitatively similar behaviour. This is consistent with findings by Daviet--Dupuis, who showed that the  Lukyanov--Zamolodchikov conjecture holds up well in a functional renormalization group analysis \cite{DavietDupuis2018}.

Finally, we study the loop expansion of the tropical sine-Gordon model and its large-order behaviour in Section~\ref{sec:largeorder}. In contrast to the mass-coupling expansion, the loop-expansion coefficients grow factorially. We compute the two- and four-point functions to 60 loops and numerically determine the leading singularity in the Borel plane. In Section~\ref{subsec:tropical-instanton}, we reproduce this number by an independent argument: at large loop order, the tropical loop equation turns into a Hamilton--Jacobi equation for a classical particle, and the Borel singularity is the classical action of a specific trajectory of this particle, which we call the \emph{tropical instanton}. The argument is heuristic, but the numerical agreement suggests that it can be made systematic. We postpone a deeper and rigorous analysis of the large-loop order behaviour of the tropical sine-Gordon model to future work.

\section{Tropical quantum field theory}
\label{sec:qfttr}
\subsection{Path integral formulation and the tropical loop equation}
\label{subsec:tropical_loop_eq} 

In this section, we briefly recall some key steps from \cite{Borinsky:2025stz}. We will work with a general 
scalar potential $V[\varphi] = \sum_{k \geq 2} \lambda_k \frac{\varphi^k}{k!}$.\footnote{The article~\cite{Borinsky:2025stz} only discussed cubic or higher interactions. Here, we also include additional quadratic terms.} 
For the sake of clarity, we note that the potential $V$ is defined with an extra negative sign compared to a standard QFT definition (see, e.g., Eqs. (7.25) and (7.27) in \cite{ZinnJustin2021}). This is to avoid an extra minus sign in the path integral weight $e^{-S}$.
The starting point for our tropical QFT analysis is to write a family of actions indexed by a real deformation parameter~$\xi >0$,
\begin{align} \label{eq:Sgeneral} S^\xi[\varphi, j] = \int_{\R^{D \cdot \xi}} \dd x^{D\cdot \xi} \left( \frac12 \varphi \,(\Box + m_0^2)^\xi \varphi - V[\varphi] - j \varphi \right) \, , \end{align}
with $\Box = - \partial_{\mu} \partial^{\mu}$, and $m_0^2$ a mass scale regulator parameter. Here, the spacetime\footnote{%
We use Euclidean conventions throughout this article. Often our statements translate verbatim to Minkowski spacetime signature; for instance, if all particles are massive or if we restrict to sufficiently generic kinematics.}
 dimension $D$ is an arbitrary parameter in the spirit of dimensional regularization \cite{tHooft:1972tcz} and the Wilson-Fisher approach \cite{Wilson:1971dc},
which we will assume to be real and non-negative.
A typical QFT with a standard Klein-Gordon propagator is obtained when $\xi =1$.
The \emph{partition functional} is formally defined by the path integral
\begin{align} Z^\xi[j] = \int \mathcal{D} [\varphi] \exp(-S^\xi[\varphi,j])\,. \end{align}
The free energy is $W^\xi[j] = \log Z^\xi[j],$ and
the effective action is %
\begin{align} \label{eq:effW} \Gamma^\xi[\varphi] = W^\xi[j] - \int \dd^{D\cdot \xi} x \, j(x) \varphi(x)\,, \end{align} with the (classical) field $\varphi$ understood to be the function of $j$, given by
\begin{equation}
\label{eq:phiofj}
    \varphi = \frac{\delta W^\xi}{\delta j} \, .
\end{equation}
As was shown in \cite{Borinsky:2025stz}, the limit $\xi \rightarrow 0^+$ exists and has a nice form. From the perspective of Eq.~\eqref{eq:Sgeneral}, one might naively expect that this limit will correspond to the well-studied 0-dimensional QFT. However, the $\xi \rightarrow 0^+$ limit is nontrivial and only coincides with 0-dimensional QFT if also $D=0$.
For general $D \geq 0$, but in the limit $\xi \rightarrow 0^+$, the effective action takes the form%
\begin{align} \label{eq:Gammaxitr} \Gamma^\xi[\varphi] = \int \dd x^{D\cdot \xi} \left( -\frac12 \varphi(x) \,(\Box + m_0^2)^{\xi} \varphi(x) + \Gamma^\tr( \varphi(x)) \right) + \text{(terms vanishing when $\xi \rightarrow 0^+$)} \, , \end{align}
where $\Gamma^\tr(\varphi)= \Gamma^\tr(D,\varphi, \lambda_2,\lambda_3, \lambda_4, \ldots)$ is the \emph{tropical effective action}, which is a function (not a functional) in the parameter $\varphi$, the coupling constants $\lambda_2,\lambda_3,\lambda_4,\ldots$
and the original spacetime dimension $D$. The tropical effective action is completely fixed by the \emph{tropical loop equation},
\begin{equation}
 \label{eq:Tr_Loop_Eq}
  \mathcal{P}_D \Gamma^{\mathrm{tr}} = 
  \bigl(1-\partial_\varphi^2\Gamma^{\mathrm{tr}}\bigr)^{-1}-1 \, ,
\end{equation}
where the differential operator $\mathcal P_D$ is given by
\begin{equation}
\label{eq:PD}
  \mathcal{P}_D
  =
  -D-\Bigl(1-\tfrac{D}{2}\Bigr)\varphi\,\partial_\varphi +\sum_{k\ge 2}\!\Bigl(k-D\bigl(\tfrac{k}{2}-1\bigr)\Bigr)\lambda_k\,\partial_{\lambda_k} \, ,
\end{equation}
and the boundary condition 
\begin{equation}
\label{eq:boundary_condition_general}
    \lim_{\hbar \rightarrow 0} \hbar \, \Gamma^{\text{tr}}(D,\hbar^{-\frac{1}{2}} \varphi, \lambda_2, \hbar^{1/2} \lambda_3,\hbar^{2/2} \lambda_4, \hbar^{3/2} \lambda_5, \ldots) = \sum_{k \geq 2} \lambda_k \frac{\varphi^k}{k!}\, .
\end{equation}

It was proven in \cite{Borinsky:2025stz} that the effective action satisfies the PDE~\eqref{eq:Tr_Loop_Eq} and the condition \eqref{eq:boundary_condition_general}, which together fix $\Gamma^\tr$ completely in the tropical limit. 

The name \emph{tropical} quantum field theory is geometric in origin: in the limit $\xi \rightarrow 0^+$, the Symanzik polynomials of each Feynman graph tropicalize, so that the graph's Feynman integral is governed by the \emph{tropical} piecewise-linear (max-plus) geometry of its Symanzik polynomials. See, for instance, \cite{maclagan2015introduction} for an introduction to tropical geometry.

The operator~\eqref{eq:PD} can be interpreted as a Callan--Symanzik operator. 
The constant contribution $-D$ accounts for the (classical) overall scaling of the effective action $\Gamma^\tr$.
Similarly, the coefficient in front of every derivative operator represents the classical scaling of each of the involved parameters.
Moreover, Eq.~\eqref{eq:Tr_Loop_Eq} allows us to connect to the realm of the \emph{functional renormalization group}. A similar equation  appears in the context of the local potential approximation of the \emph{Wetterich equation}~\cite{Wetterich:1992yh} (see, e.g.,~\cite{Codello2012}), which itself is of central interest in functional renormalization group theory.

\paragraph{A remark on the Fourier transformation}
Equation~(3) of \cite{Borinsky:2025stz} is a Fourier-transformed version of Eq.~\eqref{eq:Gammaxitr}. We briefly explain how to translate between them. Modulo an auxiliary mass scale parameter, the former equation is, 
\begin{equation}
    \frac{\delta^n \Gamma^\xi}{\delta \widetilde \varphi(p_1) \cdots \delta  \widetilde \varphi(p_n)} = \delta^{(D\cdot \xi)} \left(\sum_{i=1}^n p_i\right) \frac{\partial^n \Gamma^\tr}{\partial \varphi^n} \big|_{\varphi=0}( 1 + {o_\xi}(1))\, ,
\end{equation}
where ${o_\xi}(1)$ goes to $0$ when $\xi \rightarrow 0^+$,
and $\widetilde \varphi$  is the Fourier-transformed  field.
Inverse-Fourier transforming the $n$ momenta gives
\begin{gather} \begin{gathered} \int \frac{\dd^{D\cdot \xi} p_1}{(2\pi)^{D\cdot\xi}} e^{ip_1 x_1} \cdots \int \frac{\dd^{D\cdot \xi} p_n}{(2\pi)^{D\cdot\xi}} e^{ip_n x_n} \frac{\delta^n \Gamma^{\xi}}{\delta \widetilde \varphi(p_1) \cdots \delta \widetilde \varphi(p_n)} = \\
\left( \prod_{i=1}^{n-1} \delta^{(D\cdot \xi)}(x_i-x_n) \right) \frac{\partial^n \Gamma^{\text{tr}}}{\partial \varphi^n} \big|_{\varphi=0}( 1 + {o_\xi}(1))\,. \end{gathered} \end{gather}
Eq.~\eqref{eq:Gammaxitr} follows.

We note that the bare quadratic propagator term $\lambda_2 \varphi^2$ was not included in the expansion in Eq.~(3) of \cite{Borinsky:2025stz}.  However, the argument from \cite{Borinsky:2025stz} that derives Eq.~\eqref{eq:Tr_Loop_Eq} goes through without modification if we include this quadratic term.

We further remark that, just as ordinary QFT, tropical QFT suffers from IR divergences when $D\rightarrow 2$. For instance, the position-space propagator diverges when the mass regulator is removed. To see this, consider the
Euclidean propagator of the $\xi$-deformed QFT in position space and introduce Schwinger parameters 
\begin{gather} \begin{gathered} \int_{\R^{D\xi}} \frac{\dd^{D\xi}p}{(2\pi)^{D\xi}} \frac{e^{ipx}}{(p^2 +m_0^2)^\xi} = \frac{1}{\Gamma(\xi)} \int_0^\infty \frac{\dd \alpha}{\alpha} \alpha^\xi \int_{\R^{D\xi}} \frac{\dd^{D\xi}p}{(2\pi)^{D\xi}} \exp(ipx - \alpha(p^2+m_0^2))= \\
\frac{(4\pi)^{-D\xi/2}}{\Gamma(\xi)} \int_0^\infty \frac{\dd \alpha}{\alpha} \alpha^{\xi(1-D /2)} \exp( - \frac{x^2}{4\alpha} -\alpha m_0^2)\,. \end{gathered} \end{gather}
If $m_0^2 = 0$, $D \leq 2$ and $\xi > 0$, the integral is divergent. So, the propagator itself has an IR divergence in the massless case in low dimensions. Only in $D=2$ do we lose the (nontrivial) dependence  of the propagator on the deformation parameter $\xi$. The divergence will therefore clearly survive the tropical limit $\xi \rightarrow 0$.  We will postpone an in-depth discussion of these issues of tropicalization of $D \leq 2$ QFTs to future research. Here, we assume that $m_0^2 >0$ and that $D$ is a generic parameter before we take the tropical limit. After taking the tropical limit, we also study the limit $D\rightarrow 2$ if needed. 
\subsection{The tropical free energy and the tropical partition function}
\label{sec:freeenergy}
In \cite{Balduf:2025idm}, differential equations for the free energy and the partition function of $O(N)$-symmetric $\varphi^4$ theory in the tropical limit were derived from Eq.~\eqref{eq:Tr_Loop_Eq}.
In this section, we record the general version of these differential equations that hold for all scalar QFTs.

The tropical free energy $W^\tr(j)= W^\tr(D,j, \lambda_2, \lambda_3, \ldots)$ is the free energy in the tropical limit:
\begin{align} W^\xi[j] &= \int \dd x^{D\cdot \xi} \left(\frac12 j(x) \,(\Box + m_0^2)^{-\xi} j(x) + W^\tr( j(x)) + o_\xi(1) \right) \, . \end{align}
As before, $o_\xi(1)$ denotes terms that vanish when $\xi \rightarrow 0^+$. By the relationship~\eqref{eq:effW} to the effective action, we find
\begin{align} \Gamma^\tr \Big|_{\varphi = j + \frac{\partial W^\tr}{\partial j} } = W^\tr + \frac12 \left( \frac{ \partial W^\tr}{\partial j } \right)^2. \end{align}
Here, we used that $(\Box + m_0^2)^{\xi} \rightarrow 1$, when we evaluate it on propagators whose value is fixed by the external kinematics. %
In other words, the Feynman rules in the tropical limit
set all bridge propagators (i.e.~propagators which carry no nontrivial loop momentum) of a Feynman graph to 1.\footnote{Even though bridge propagators do not contribute, propagators with loop momentum flowing through them contribute nontrivially via the tropical Feynman rules (i.e.~the Hepp bound recursion).
So, $(\Box + m_0^2)^{\xi} \neq 1$ when computing virtual corrections.}
Inverting \eqref{eq:phiofj} to obtain $j$ as a function of $\varphi$ and differentiating \eqref{eq:effW} with respect to $\varphi$ we further obtain in the limit $\xi \rightarrow 0^+$,
\begin{align} \frac{ \partial \Gamma^\tr}{ \partial \varphi} \Biggr|_{\varphi = j + \frac{\partial W^\tr}{\partial j} } &= \frac{\partial W^\tr}{\partial j} \, , & \frac{ \partial^2 \Gamma^\tr}{ \partial \varphi^2} \Biggr|_{\varphi = j + \frac{\partial W^\tr}{\partial j} } &= 1- \frac{1}{1+\frac{\partial^2 W^\tr}{\partial j^2}} \, . \end{align}
From this and Eq.~\eqref{eq:Tr_Loop_Eq}, we get the tropical loop equation for the free energy\footnote{Here we use 
\begin{equation}
    \partial_{\lambda_k} \Gamma^{\xi} = \partial_{\lambda_{k}} W^{\xi} \Bigr|_{j = j(\varphi)} \Rightarrow \partial_{\lambda_k} \Gamma^\tr \Bigr|_{\varphi = j + \frac{\partial W^\tr}{\partial j} }= \partial_{\lambda_{k}} W^\tr \, ,  
\end{equation}
with $j(\varphi)$ obtained by inverting \eqref{eq:phiofj}. }
\begin{gather} \begin{gathered} \label{eq:Tr_Loop_Eq_W} \sum_{k\ge 2}\!\Bigl(k-D\bigl(\tfrac{k}{2}-1\bigr)\Bigr)\lambda_k\,\partial_{\lambda_k} W^\tr -D W^\tr- \left(\frac{\partial W^\tr}{\partial j}\right)^2 + \frac{D-2}{2} j \frac{\partial W^\tr}{\partial j} = \frac{\partial^2 W^\tr}{\partial j^2} \, . \end{gathered} \end{gather}
This PDE has the same form as the Polchinski equation in the functional renormalization group framework \cite{Polchinski:1983gv}. 
Felder proved rigorously that this precise equation holds in the  local potential approximation of QFT~\cite{Felder1987}.

Finally, passing to the partition function $Z^\tr = \exp( W^\tr)$ gives the PDE,
\begin{gather} \begin{gathered} \sum_{k\ge 2}\!\Bigl(k-D\bigl(\tfrac{k}{2}-1\bigr)\Bigr)\lambda_k\,\partial_{\lambda_k} Z^\tr -D Z^\tr \log Z^\tr + \frac{D-2}{2} j \frac{\partial Z^\tr}{\partial j} = \frac{\partial^2 Z^\tr}{\partial j^2}\,. \end{gathered} \end{gather}
In the $D \rightarrow 0$ limit, this PDE becomes linear. In fact, it reduces to the 
standard PDE satisfied by $0$-dimensional QFT. As pointed out in \cite{Balduf:2025idm}, 
this shows that tropical QFT is a non-linear deformation of $0$-dimensional QFT: 
the non-linearity stems entirely from the $\log$ term.

\subsection{Tropical perturbation theory}
\label{sec:perturbation}
As in traditional QFTs, the partition function, the free energy, and the effective action of the $\xi$-deformed QFT can be expanded perturbatively using Feynman graphs. For example, the $\xi$-deformed free energy expands as
\begin{align} \label{eq:Wnontrop} \frac{\delta^n W^\xi}{\delta \widetilde j(p_1) \cdots \delta \widetilde j(p_n)} = (2\pi)^{D\cdot \xi} \delta^{(D\cdot \xi)}\left( \sum_{i=1}^n p_i \right) \sum_{ \substack{G, n(G) = n } } \frac{\prod_{v\in V_G} \lambda_{|v|}}{|\Aut(G)|} I^\xi_G(p_1,\ldots,p_n)\,, \end{align}
where we sum over all connected graphs $G$ with $n$ legs,
and arbitrarily many vertices of degree $2$ or higher. 
The term $\prod_{v\in V_G} \lambda_{|v|}$ is the coupling monomial 
associated with $G$ and $|\Aut(G)|$ is the symmetry factor.
The standard formula for the Feynman integral reads~(see, e.g., \cite{Weinzierl:2022eaz}),
\begin{align*} I^\xi_G(p_1,\ldots,p_n) &= (2\pi)^{-L(G)\cdot D\cdot \xi} \int \frac{\dd^{D\cdot \xi} k_1 \cdots \dd^{D\cdot \xi} k_{L(G)}} {\prod_{e\in E_G} (q_e^2+m_0^2)^{\xi}}\,, \end{align*}
where we integrate over the $L(G)$ \emph{loop} momenta of $G$ 
and $q_e$ is the total momentum flowing through the edge $e$.

\subsection{Computing tropical Feynman integrals}
\label{sec:tropfeynman}
In the tropical limit $\xi \rightarrow 0^+$, the Feynman integral 
tropicalizes to the \emph{Hepp bound}, which was first studied in \cite{Panzer:2019yxl}.
The Hepp bound takes the role of the tropicalized Feynman integral, because
\begin{align} \label{eq:hepplimit} \lim_{\xi \rightarrow 0^+} I^\xi_G(p_1,\ldots,p_n) = \HB_D(G)\,. \end{align}
See, e.g., Proposition~5 of \cite{Borinsky:2025stz} for an explicit proof.
An easy way to compute the tropical Feynman integral, i.e.~the Hepp bound, is the recursive formula
\begin{equation}
\label{eq:hepp_recursion}
  \HB_D(G) = \frac{1}{\omega_D(G)} \sum_{e\in E_G} \HB_D(G\setminus e) \, ,
\end{equation}
where 
$\omega_D(G) = |E_G|-\frac{D}{2} L(G)$ and one assumes the terminal condition $\HB_D( \emptyset ) =1$ where $\emptyset$ stands for any graph without edges. The tropical Feynman integral is a rational function in $D$, which has a pole whenever the original, non-tropical Feynman integral has a logarithmic UV singularity.

We emphasize that the tropical Feynman integral does not depend on the external kinematics $m_0^2, p_1,\ldots,p_n$ anymore, provided that $m_0^2 >0$ such that there are no IR singularities. The loss of dependence is the per-diagram manifestation of the homogenization phenomenon of the effective action in the tropical limit in Eq.~\eqref{eq:Gammaxitr}. IR singularities generally depend on the kinematic configuration. As singular structures survive the tropical limit, tropical theories with nontrivial IR behaviour will exhibit dependence on the kinematics.

Tropical QFT quantities can be expanded in terms of the tropical Feynman integral
analogously to the usual QFT expansions. For instance, we get by combining Eqs.~\eqref{eq:Wnontrop} with~\eqref{eq:hepplimit} while stripping off the overall momentum conservation,
\begin{align} \label{eq:FuncGraphExpanded} \frac{\partial^n W^\tr}{\partial j^n} \Biggr|_{j=0} = \sum_{ \substack{G, n(G) = n } } \frac{\prod_{v\in V_G} \lambda_{|v|}}{|\Aut(G)|} \HB_D(G) \,, \end{align}
where we sum over all connected graphs $G$ with $n$ legs.

Tropical QFT inherits many nontrivial features of traditional QFT, but 
has the key advantage of being much more computationally accessible. Hence, tropical 
QFT offers itself as an instructive testbench for QFT computations. In the remainder of this section, we will demonstrate this accessibility explicitly by computing a couple of Feynman diagrams. In traditional QFTs, such computations would require much more advanced technology.

One can show that the tropical Feynman integral respects the factorization identity,
\begin{align} \label{eq:heppfactor} \mathcal H_D\left( \begin{tikzpicture}[every path/.style={thick},baseline={(0,-.1)}] \def\rad{.6} \def\rud{.6} \def\radl{0} \coordinate (v1) at (-\rud,0); \coordinate (v2) at (\rud,0); \coordinate (v01) at (-\radl,0); \coordinate (v02) at (\radl,0); \filldraw[pattern color=col3,pattern=north east lines] (v1) circle(\rad); \filldraw[pattern color=col3,pattern=north east lines] (v2) circle(\rad); \fill (v01) circle(2pt); \fill (v02) circle(2pt); \node[fill=white,circle] at (v1) {$G_1$}; \node[fill=white,circle] at (v2) {$G_2$}; \end{tikzpicture} \right) = \mathcal H_D\left( \begin{tikzpicture}[every path/.style={thick},baseline={(0,-.1)}] \def\rad{.6} \def\rud{.8} \def\radl{.2} \coordinate (v1) at (-\rud,0); \coordinate (v2) at (\rud,0); \coordinate (v01) at (-\radl,0); \coordinate (v02) at (\radl,0); \filldraw[pattern color=col3,pattern=north east lines] (v1) circle(\rad); \filldraw[pattern color=col3,pattern=north east lines] (v2) circle(\rad); \fill (v01) circle(2pt); \fill (v02) circle(2pt); \node[fill=white,circle] at (v1) {$G_1$}; \node[fill=white,circle] at (v2) {$G_2$}; \end{tikzpicture} \right) = \mathcal H_D\left( \begin{tikzpicture}[every path/.style={thick},baseline={(0,-.1)}] \def\rad{.6} \def\rud{.8} \def\radl{.2} \coordinate (v1) at (-\rud,0); \coordinate (v01) at (-\radl,0); \filldraw[pattern color=col3,pattern=north east lines] (v1) circle(\rad); \fill (v01) circle(2pt); \node[fill=white,circle] at (v1) {$G_1$}; \end{tikzpicture} \right) \cdot \mathcal H_D\left( \begin{tikzpicture}[every path/.style={thick},baseline={(0,-.1)}] \def\rad{.6} \def\rud{.8} \def\radl{.2} \coordinate (v2) at (\rud,0); \coordinate (v02) at (\radl,0); \filldraw[pattern color=col3,pattern=north east lines] (v2) circle(\rad); \fill (v02) circle(2pt); \node[fill=white,circle] at (v2) {$G_2$}; \end{tikzpicture} \right)\,, \end{align}
which follows from the recursive definition~\eqref{eq:hepp_recursion} (see, e.g., \cite[Proposition~4]{Borinsky:2025stz}). 

We can compute the tropical Feynman integral of a single propagator graph using 
the recursion~\eqref{eq:hepp_recursion}
as a consistency check. We find that 
\begin{align} \HB_D \left( \begin{tikzpicture}[baseline={(0,-.1)}] \def\rud{.4} \def\radl{.6} \def\rad{.2} \coordinate (TR) at ( \rud, 0); \coordinate (TL) at (-\rud, 0); \draw (-\radl,-\rad) -- (TL); \draw ( \radl,-\rad) -- (TR); \draw ( \radl, \rad) -- (TR); \draw (-\radl, \rad) -- (TL); \draw (TR) -- (TL); \foreach \v in { TR, TL} \fill (\v) circle (1.6pt); \end{tikzpicture} \right) = \frac{1}{ \omega_D \left( \begin{tikzpicture}[baseline={(0,-.1)}] \def\rud{.12} \def\radl{.2} \def\rad{.04} \coordinate (TR) at ( \rud, 0); \coordinate (TL) at (-\rud, 0); \draw (-\radl,-\rad) -- (TL); \draw ( \radl,-\rad) -- (TR); \draw ( \radl, \rad) -- (TR); \draw (-\radl, \rad) -- (TL); \draw (TR) -- (TL); \foreach \v in { TR, TL} \fill (\v) circle (1pt); \end{tikzpicture} \right) } \HB_D \left( \begin{tikzpicture}[baseline={(0,-.1)}] \def\rud{.4} \def\radl{.6} \def\rad{.2} \coordinate (TR) at ( \rud, 0); \coordinate (TL) at (-\rud, 0); \draw (-\radl,-\rad) -- (TL); \draw ( \radl,-\rad) -- (TR); \draw ( \radl, \rad) -- (TR); \draw (-\radl, \rad) -- (TL); \draw (TR) -- (TL); \fill[white] (0,0) circle (.15); \foreach \v in { TR, TL} \fill (\v) circle (1.6pt); \end{tikzpicture} \right) = \frac{1}{ 1 - \frac{D}{2} \cdot 0} = 1 \,, \end{align}
where we used that the tropical Feynman integral of a graph without edges is $1$.

Slightly less trivially, we find for the box Feynman graph
\begin{align*} \HB_D \left( \begin{tikzpicture}[baseline={(0,-.1)}] \def\rud{.5} \def\radl{.7} \coordinate (BL) at (-\rud,-\rud); \coordinate (BR) at ( \rud,-\rud); \coordinate (TR) at ( \rud, \rud); \coordinate (TL) at (-\rud, \rud); \draw (-\radl,-\radl) -- (BL); \draw ( \radl,-\radl) -- (BR); \draw ( \radl, \radl) -- (TR); \draw (-\radl, \radl) -- (TL); \draw (BL) -- (BR); \draw (BR) -- (TR); \draw (TR) -- (TL); \draw (TL) -- (BL); \foreach \v in {BL, BR, TR, TL} \fill (\v) circle (1.6pt); \end{tikzpicture} \right) = \frac{1} {\omega_D \left( \begin{tikzpicture}[baseline={(0,-.1)}] \def\rud{.12} \def\radl{.2} \coordinate (BL) at (-\rud,-\rud); \coordinate (BR) at ( \rud,-\rud); \coordinate (TR) at ( \rud, \rud); \coordinate (TL) at (-\rud, \rud); \draw (-\radl,-\radl) -- (BL); \draw ( \radl,-\radl) -- (BR); \draw ( \radl, \radl) -- (TR); \draw (-\radl, \radl) -- (TL); \draw (BL) -- (BR); \draw (BR) -- (TR); \draw (TR) -- (TL); \draw (TL) -- (BL); \foreach \v in {BL, BR, TR, TL} \fill (\v) circle (1pt); \end{tikzpicture} \right) } \cdot 4 \cdot \HB_D \left( \begin{tikzpicture}[baseline={(0,-.1)}] \def\rud{.5} \def\radl{.7} \coordinate (BL) at (-\rud,-\rud); \coordinate (BR) at ( \rud,-\rud); \coordinate (TR) at ( \rud, \rud); \coordinate (TL) at (-\rud, \rud); \draw (-\radl,-\radl) -- (BL); \draw ( \radl,-\radl) -- (BR); \draw ( \radl, \radl) -- (TR); \draw (-\radl, \radl) -- (TL); \draw (BL) -- (BR); \draw (BR) -- (TR); \draw (TL) -- (BL); \foreach \v in {BL, BR, TR, TL} \fill (\v) circle (1.6pt); \end{tikzpicture} \right) = \frac{4}{4-\frac{D}{2} \cdot 1} \cdot \HB_D \left( \begin{tikzpicture}[baseline={(0,-.1)}] \def\rud{.4} \def\radl{.6} \def\rad{.2} \coordinate (TR) at ( \rud, 0); \coordinate (TL) at (-\rud, 0); \draw (-\radl,-\rad) -- (TL); \draw ( \radl,-\rad) -- (TR); \draw ( \radl, \rad) -- (TR); \draw (-\radl, \rad) -- (TL); \draw (TR) -- (TL); \foreach \v in { TR, TL} \fill (\v) circle (1.6pt); \end{tikzpicture} \right)^3 = \frac{4}{4-\frac{D}{2}}\,, \end{align*}
where we used the factorization property~\eqref{eq:heppfactor} in the second step. 

More generally, let $C_n$ be the cycle graph with $n$ edges/propagators. Analogously 
to the computation above, we find that 
\begin{align} \label{eq:HeppCycle} \HB_D(C_n) = \frac{1}{n-\frac{D}{2}} \sum_{e\in E_{C_n}} \HB_D(C_n\setminus e) = \frac{n}{n-\frac{D}{2}} \, . \end{align}

At two loops there are only two relevant families of Feynman graphs: every two-loop 1PI graph is either two one-loop tadpoles joined at a vertex, which factorizes by Eq.~\eqref{eq:heppfactor} into $\HB_D(C_a)\HB_D(C_b)$, or a theta graph 
\begin{align} \Theta_{a,b,c} = \begin{tikzpicture}[baseline={(0,-.1)},scale=0.62,line cap=round,line join=round] \coordinate (L) at (-2.0,0); \coordinate (R) at ( 2.0,0); \foreach \ang/\lab in {60/a, 0/b, -60/c} { \draw (L) to[bend left=\ang,looseness=1.3] coordinate[pos=0.25] (u) node[pos=0.5,fill=white,inner sep=2pt] {$\lab$} coordinate[pos=0.75] (v) (R); \fill (u) circle (1.8pt); \fill (v) circle (1.8pt); } \filldraw[draw=black,thick] (L) circle (2.4pt); \filldraw[draw=black,thick] (R) circle (2.4pt); \end{tikzpicture} \, , \end{align}
 whose two vertices are joined by three paths with $a$, $b$ and $c$ edges chained together respectively. In the latter case $|E_{\Theta_{a,b,c}}|=a+b+c$, and deleting an edge from one of the paths leaves the one-loop graph formed by the other two up to bridge edges, which is therefore a one-loop graph computed in~\eqref{eq:HeppCycle}. Hence~\eqref{eq:hepp_recursion} gives
\begin{align} \begin{aligned} \label{eq:HeppTheta} \HB_D \left( \Theta_{a,b,c} \right) &= \frac{a \, \HB_D(C_{b+c}) + b \, \HB_D(C_{a+c}) + c \, \HB_D(C_{a+b})}{a+b+c-D} \, \\
&= \frac{(ab+bc+ca)\,(a+b+c-\tfrac{D}{2})\,(a+b+c-D) - abc\,(a+b+c-\tfrac{3D}{2})} {(a+b-\frac{D}{2})\,(b+c-\frac{D}{2})\,(a+c-\frac{D}{2})\,(a+b+c-D)} \, , \end{aligned} \end{align}
where we used the expression for $\HB_D(C_n)$ in Eq.~\eqref{eq:HeppCycle}. The simplest case is the $\Theta_{1,1,1}$,
\begin{align} \label{eq:HeppBananaThree} \HB_D \left(\Theta_{1,1,1} \right) = \HB_D \left( \begin{tikzpicture}[baseline={(0,-.1)}] \def\rud{.5} \coordinate (L) at (-\rud,0); \coordinate (R) at ( \rud,0); \draw (L) -- (R); \draw (L) to[bend left=60,looseness=1.3] (R); \draw (L) to[bend right=60,looseness=1.3] (R); \foreach \v in {L,R} \fill (\v) circle (1.6pt); \end{tikzpicture} \right) = \frac{3}{3-D}\,\frac{2}{2-\frac{D}{2}} \, . \end{align}
\subsection{Tropical QFTs' renormalization group fixed points}
\label{sec:trop_ren_group_fix}

In the tropical limit, the partition function, free energy, and effective action are fixed by the tropical loop equation~\eqref{eq:Tr_Loop_Eq} up to boundary conditions which, in turn, are fixed by the Lagrangian. Therefore, in contrast to most traditional QFTs, we do not need to rely on perturbation theory to say something about these quantities. Non-perturbative statements about QFT can be made in the tropical limit. Besides the application of tropical QFT to explicit evaluation of perturbative quantities in the original non-tropical QFT using global tropical sampling, this aspect of QFT tropicalization might be one of the most interesting. Here, we demonstrate the utility of the tropical approach for making such statements.

Conveniently, many of the equations governing QFTs in the tropical limit have been studied extensively in the context of the local potential approximation and the (functional) renormalization group approach (see, e.g., \cite{Felder1987,Morris:1994jc,Litim:2001up,Codello2012}). So, this section also gives a brief glimpse of the possibilities for importing results from these earlier investigations into the tropical QFT framework.

It is interesting to look at QFTs that fulfill the \emph{scale-invariant} condition $\lambda_k \partial_{\lambda_k} \Gamma^\tr = \lambda_k \partial_{\lambda_k} W^\tr = 0$ for all $k$.
This implies that we are at a renormalization group fixed point in theory space. Such fixed points 
give rise to \emph{universality classes} of physical systems and 
play a pivotal role across many parts of theoretical physics. 

It is easy to see that, in the tropical limit, QFTs obeying the scale-invariance condition have free energies that are solutions of the following reduction of~\eqref{eq:Tr_Loop_Eq_W},
\begin{gather} \begin{gathered} \label{eq:Tr_Loop_Eq_W_scalefree} -D W^\tr- \left(\frac{\partial W^\tr}{\partial j}\right)^2 + \frac{D-2}{2} j \frac{\partial W^\tr}{\partial j} - \frac{\partial^2 W^\tr}{\partial j^2} =0 \, . \end{gathered} \end{gather}
Felder~\cite{Felder1987} performed a rigorous study of this equation in the context of the local potential approximation.
By Felder's Theorem 3.2, Eq.~\eqref{eq:Tr_Loop_Eq_W_scalefree} has an even global solution $W^\tr_{2n}(j)$ for each $n\in\{2,3,\ldots\}$ where $2<D< 2+\frac{2}{n-1}$. The solution grows to $\infty$ for $j \rightarrow \pm \infty$ and has $2n-1$ critical points. Based on this fact, Felder associates the solution $W^\tr_{2n}(j)$ with a $\phi^{2n}$-type fixed point. Most prominently, there is the Wilson--Fisher fixed point for $D=3$ and $n=2$.

Similarly, the effective potential solves the following specialization of Equation~\eqref{eq:Tr_Loop_Eq} in the scale-invariant case,
\begin{align} \label{eq:troprgfixedpoint} -D\Gamma^\tr-\Bigl(1-\tfrac{D}{2}\Bigr)\varphi\,\partial_\varphi \Gamma^\tr = \bigl(1-\partial_\varphi^2\Gamma^{\tr}\bigr)^{-1}-1. \end{align}
This ODE was also studied extensively in the context of the local potential approximation.

Codello~\cite{Codello2012}  scanned the space of solutions  of \eqref{eq:troprgfixedpoint},
which yield an effective potential $\Gamma^\tr$ that stays finite for all values of the field strength $\varphi$. Codello's results are consistent with Felder's theorem
and also single out special finite solutions such as the Wilson--Fisher fixed point at $D=3$.

The two-dimensional case is  particularly nice.
For $D=2$, Eq.~\eqref{eq:troprgfixedpoint} becomes an autonomous ODE which may be written as
\begin{align} \label{eq:newton} \partial_\varphi^2\Gamma^{\tr} = - \frac{\partial}{\partial \Gamma^\tr} \left( -\Gamma^\tr- \frac{1}{2} \log\left| 1-2 \Gamma^\tr \right| \right)\, . \end{align}
The potential $V(x) := -x-\frac{1}{2} \log|1-2 x|$ is bounded from below and convex on the interval $(-\infty, \frac{1}{2} )$ (see Figure~\ref{fig:potential}). Moreover, it approaches $+\infty$ on both endpoints. Therefore, in this interval, we may interpret Eq.~\eqref{eq:newton} as Newton's equation for a particle oscillating in the potential $V(x)$. 
\begin{figure}[t]
\centering
\includegraphics[width=0.94\linewidth]{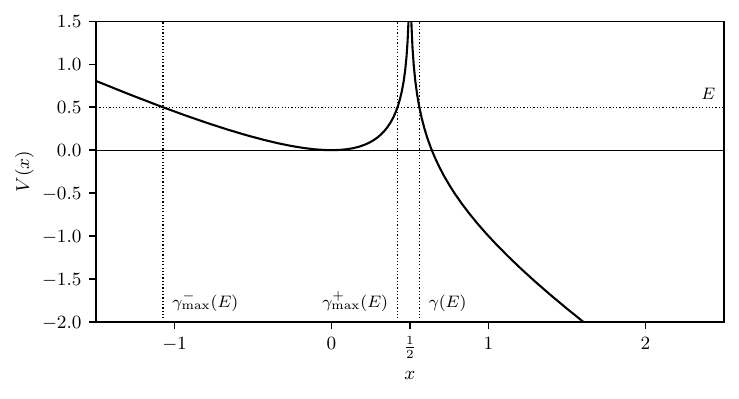}
\caption{The potential $V(x) = -x - \frac12 \log|1-2x|$.}
\label{fig:potential}
\end{figure}
We get a first, one-parameter family of real, oscillating solutions for 
$\Gamma^\tr$ that are indexed by the energy-like quantity
\begin{align} \label{eq:energyGamma} E = \frac{1}{2} \left( \partial_\varphi\Gamma^{\tr} \right)^2 -\Gamma^\tr- \frac{1}{2} \log\left| 1-2 \Gamma^\tr \right| \geq 0 \,. \end{align}

As the potential is convex in the interval $(-\infty, \frac{1}{2} )$ and grows to $+\infty$ on both boundaries, the locations $\gamma^\pm_\mathrm{max}$ of the oscillation's turning points 
are the two real solutions of $E = - \gamma^{\pm}_\mathrm{max} -\frac12 \log\left| 1-2 \gamma^{\pm}_\mathrm{max} \right|$. 
These locations are functions of $E$ and they are depicted as dotted lines in Figure~\ref{fig:potential}.
They can be expressed explicitly as $\gamma^\pm_\mathrm{max}(E) = \frac12 + \frac12 W_\pm( - e^{-2 E -1} )$, where $W_+$ is the principal branch 
and $W_-$ the other real branch of the Lambert $W$ function (see, e.g., \cite{corless1996lambert}). 

Later in this article, the \emph{period} of these oscillatory solutions will be of interest.
A standard chain of arguments using~\eqref{eq:energyGamma} gives the following integral formula for the period
\begin{align} T(E) = \sqrt{2} \int_{\gamma_\mathrm{max}^-(E)}^{\gamma_\mathrm{max}^+(E)} \frac{\dd \Gamma}{\sqrt{ E + \Gamma + \frac12 \log|1-2\Gamma| }} \, . \end{align}
The period $T(E)$ is a monotonically increasing function in $E$. This can be proven using Schaaf's theorem~\cite{MR814016}:
for $V(x) = -x - \frac12 \log(1-2x)$ the inequalities $V''(x) >0 $ and $5V'''(x)^2 - 3 V''(x) V''''(x) = 32/ (1-2x)^6 > 0$ 
on the interval $x \in (-\infty, \frac12)$ imply that $T'(E) > 0$ for all $E >0$.

Note further that, for small $E$, \eqref{eq:newton} is solved by
\begin{align} \label{eq:phase_transition_point} \Gamma^\tr = \sqrt{E} \cos(\sqrt{2} \varphi + \varphi_0) + \mathcal O(E)\,. \end{align}
Imposing the $\mathbb Z/2 \mathbb Z$ symmetry $\Gamma^\tr(\varphi) = \Gamma^\tr(-\varphi)$ further fixes the shift parameter $\varphi_0=0$.
Moreover, we learn that the limiting minimal period is $T(0) = 2\pi / \sqrt{2} = \sqrt{2} \pi$.

The (effective) action we obtain is a specific instance of the sine-Gordon action. 
The finding that $D=2$ tropical QFT in the illustrated small-$E$ limit 
singles out the sine-Gordon potential motivates our deeper dive into
this model in the second part of this article starting with the next section.

In the interval $(\frac12,\infty)$, the potential $V(x)$ is monotonically decreasing. In this 
interval, the real solutions to~\eqref{eq:newton} are also easily classified. They are indexed by 
a (not necessarily positive) energy $E$. The solution $\Gamma^\tr(\varphi)$ decreases monotonically from $\infty$ to a turning point/minimum at $\gamma(E)$,
then increases monotonically to $\infty$. 
The turning point is $\gamma(E) = \frac12 + \frac12 W( e^{-2 E -1} )$ (see Figure~\ref{fig:potential}), where we use the principal branch of the Lambert $W$ function.
As before, we have a shift symmetry, so we can arbitrarily fix the time of the turning point to $\varphi=0$. 
Hence, we get a second family of $\Z/2\Z$-symmetric, but non-periodic potentials indexed by values $E \in \R$.

Therefore, we find two continuous families of scale-invariant tropical QFTs in $D=2$.
From a purely physical perspective, the periodic family has a fundamental difference to the second, non-periodic solution:
by construction we have $\Gamma^\tr < \frac12$ for that family. It follows that $1- \partial_\varphi^2 \Gamma^\tr = \frac{1}{1-2\Gamma^\tr} > 0$,
so the full propagator $1/(1- \partial_\varphi^2 \Gamma^\tr)$ is positive. Positivity of the propagator is a necessary condition for an associated
measure on the space of field configurations to exist~\cite{Glimm:1987ylb}. 
The other family of solutions that fulfill $\Gamma^\tr > \frac12$ leads to a negative propagator.
It is therefore not possible to interpret these solutions as emerging from Euclidean scalar QFTs.

Later in this article, we will show that the first family of periodic scale-invariant tropical QFTs 
can be approached by running the coupling of a bare, tropical sine-Gordon theory into its IR fixed-point.
We do not have a similar interpretation for the second, continuous family of scale-invariant solutions.
\section{\texorpdfstring{$D=2$}{D=2} tropical QFT and tropical sine-Gordon QFT}
\label{sec:sinegordon}
\subsection{The tropical loop equation for sine-Gordon QFT}

The sine-Gordon model is a two-dimensional quantum field theory of a bosonic field. It is given by the action functional, 
\begin{align} \label{eq:S} S_{sG}[\varphi,j] = \int_{\R^{D}} \left( \frac12 \varphi \,(\Box + m_0^2) \varphi + \frac{\alpha}{g^2} (1-\cos(g \varphi)) - j \varphi \right) \dd^D x\,, \end{align}
where we follow Coleman's convention in his seminal article \cite{Coleman1975}.\footnote{We rename $\beta \rightarrow g$ to avoid the confusion with the $\beta$-functions.} Just as in Coleman's analysis, we add an auxiliary mass parameter $m_0^2$ acting as an IR regulator, and we eventually take $m_0^2 \rightarrow 0^+$.
Coleman showed that if $g^2$ exceeds $8\pi$, the energy density of the theory is unbounded from below.
If $g^2=4\pi$, the theory is equivalent to the zero-charge sector of free massive Fermi field theory. For other values of $g$, the theory is equivalent to the zero-charge sector of the massive Thirring model. In the following, from the analysis of the $\beta$-functions, we identify the tropical counterpart of the \emph{Coleman point} $g^2 = 8 \pi$ to be $g^2 = 2$. We leave the study of the relation between various special values of $g$ in both theories, such as the point $g^2=4 \pi$, to future work. 

The tropicalization procedure introduced in \cite{Borinsky:2025stz} can be directly applied to Eq.~\eqref{eq:S}. First, we define the $\xi$-deformed sine-Gordon action, 
\begin{align} \label{eq:Sxi} S^\xi_{sG}[\varphi,j] = \int_{\R^{D \cdot \xi}} \left( \frac12 \varphi \,(\Box + m_0^2)^\xi \varphi + \frac{\alpha}{g^2} (1-\cos(g \varphi)) - j \varphi \right) \dd^{D\cdot \xi} x\,, \end{align}
where the dimension is rescaled as $D \mapsto D \cdot \xi$ and the propagator is powered 
$\Box + m_0^2 \mapsto (\Box + m_0^2)^\xi$ using a deformation parameter $\xi > 0$. 
Obviously, we get the original sine-Gordon action when $\xi = 1$.
The tropical limit is attained when $\xi \rightarrow 0^+$.
The parameter $\alpha$ is the bare mass squared of the non-tropicalized sine-Gordon model.
The reparameterization $\varphi \rightarrow {{\varphi}}/{g}$ and $g \rightarrow \sqrt{\hbar}$
recovers the typical loop-counting parameter $\hbar$, as the tree-level action acquires the overall factor $\frac{1}{\hbar}$.

Note that the potential in~\eqref{eq:Sxi} is equivalent to the substitutions of the coupling constants
$\lambda_{2k} = -\alpha (-g^2)^{k-1}$ and $\lambda_{2k+1} = 0$ in Eq.~\eqref{eq:PD}.
So, the operator $\mathcal{P}_D$ in Eq.~\eqref{eq:PD} specializes to,
\begin{equation}
\label{eq:PDsG}
    \mathcal{P}_D = -D-\Bigl(1-\tfrac{D}{2}\Bigr)\varphi\,\partial_\varphi
  +\Big(1-\tfrac{D}{2} \Bigr) g \partial_{g} + 2 \, \alpha \partial_{\alpha}  \, .
\end{equation}
In fact, this form of $\mathcal P_D$ holds for \textit{any} potential $V[\varphi] = \frac{\alpha}{g^2} F[g \varphi]$ where $F$ is some function.

Further, Eq.~\eqref{eq:Tr_Loop_Eq} takes the form 
\begin{align} \label{eq:Tr_Loop_Eq_sG} \left( -D-\Bigl(1-\tfrac{D}{2}\Bigr)\varphi\,\partial_\varphi +\Big(1-\tfrac{D}{2} \Bigr) g \partial_{g} + 2 \, \alpha \partial_{\alpha} \right) \Gamma^{\mathrm{tr}}= \bigl(1-\partial_\varphi^2\Gamma^{\mathrm{tr}}\bigr)^{-1}-1 \, , \end{align}
and the boundary condition~\eqref{eq:boundary_condition_general} becomes, 
\begin{equation}
\label{eq:Boundary_Cond}
    \lim_{\hbar \rightarrow 0^+} \hbar \, \Gamma^{\text{tr}}(D,\hbar^{-\frac{1}{2}} \varphi, \alpha, \hbar^{\frac{1}{2}} g) =  -\frac{\alpha}{g^2} \left( 1-\cos{g \varphi}\right)  \, .
\end{equation}
It is a mechanical exercise to solve \eqref{eq:Tr_Loop_Eq_sG} and \eqref{eq:Boundary_Cond} perturbatively for the first few orders in $\alpha,g$ and $\varphi$ via 
an Ansatz (see \S3.3 of \cite{Borinsky:2025stz} for a prescription using a fixed-point method).
The power series solution of the PDE, which respects the boundary condition, starts as
\begin{align} \begin{aligned} \label{eq:expansion} \Gamma^{\mathrm{tr}}(D,\varphi,\alpha,g) = \,& -\frac{\alpha \varphi^{2}}{2} + \frac{\alpha g^{2} \varphi^{4}}{4!} -\frac{\alpha g^{4} \varphi^{6}}{6!} + \frac{\alpha g^{6} \varphi^{8}}{8!} +\ldots \\
& -\frac{\alpha g^{2} \varphi^{2}}{2(D-2)} + \frac{\alpha g^{4} \varphi^{4}}{4!(D-2)} - \frac{\alpha g^{6} \varphi^{6}}{6!(D-2)} +\ldots \\
& +\frac{\alpha^2 g^2 \varphi^2}{D-4} - \frac{\alpha^2 g^4 \varphi^4}{3(D-4)} \\
&- \frac{\alpha g^{4} \varphi^{2}}{4(D-2)^{2}}+ \frac{\alpha g^{6} \varphi^{4}}{48(D-2)^{2}}+\ldots \\
& + \frac{(3D-8)\,\alpha^{2}g^4\varphi^{2}}{(D-4)(D-3)(D-2)} - \frac{\alpha g^{6}\varphi^2}{12(D-2)^{3}} +\ldots\\
 & + \frac{\alpha}{D-2} + \frac{\alpha g^2}{2(D-2)^2} + \frac{\alpha^2}{4-D} + \ldots \\
& + \ldots \, . \end{aligned} \end{align}

\subsection{Perturbation theory of tropical sine-Gordon theory}
\label{sec:pertsinegordon}

The terms in the expansion~\eqref{eq:expansion} can also be computed via the traditional route: write down all Feynman diagrams,  evaluate them using the tropical Feynman rules from Section~\ref{sec:perturbation}, and sum them all up. 
The first line of~\eqref{eq:expansion} features only tree-level, vertex-type diagrams, whose shape and value can be read off from the potential in \eqref{eq:Sxi} using 
\begin{align} \frac{\alpha}{g^2} (1-\cos(g \varphi)) = -\alpha \sum_{k \geq 1} (-1)^k \frac{g^{2(k-1)} \varphi^{2k}}{(2k)!}\,\,, \end{align}
and therefore 
\begin{equation}
\label{eq:vertices}
        \vertextwo = -\alpha \, , \quad \vertexfour =  \alpha g^2 \, , \quad \ldots
\end{equation}
The second line of~\eqref{eq:expansion} contains the contributions from one-loop tadpole diagrams.
These diagrams are divergent for $D \rightarrow 2$.  The one-loop tadpoles are just the $1$-cycle diagrams $C_1 = \vertexloop$. 
Eq.~\eqref{eq:HeppCycle} gives their tropical Feynman integral. The 
external legs do not influence the (tropical-)analytic part of the Feynman integral. Only the coupling factor changes 
depending on the vertex structure. So,
\begin{align} \vertextwoloop = \frac12 \alpha g^2 \frac{1}{1-\frac{D}{2}}\,, \quad \vertextwolooptwolegs = -\frac12 \alpha g^4 \frac{1}{1-\frac{D}{2}} \, ,\quad \ldots \end{align}
The contributions of the bubble (double-edge) diagrams to the third line of~\eqref{eq:expansion} are
\begin{align} \bubbletwolegs &= -\frac12 \alpha^2 g^2 \frac{2}{2-\frac{D}{2}}\,, \\
\bubblefourlegs + 3\,\bubbletwotwolegs &= \left(\frac12+\frac32\right)\alpha^2 g^4 \frac{2}{2-\frac{D}{2}} = 2\alpha^2 g^4 \frac{2}{2-\frac{D}{2}}\,. \end{align}
The fourth line of~\eqref{eq:expansion} collects the two-loop diagrams with two tadpoles attached to a single vertex, whose analytic part is the square of the tadpole one,
\begin{align} \vertextwolooptwo = -\frac18 \alpha g^4 \frac{1}{\left(1-\frac{D}{2}\right)^{2}}\,, \quad \vertextwolooptwotwolegs = \frac18 \alpha g^6 \frac{1}{\left(1-\frac{D}{2}\right)^{2}} \, ,\quad \ldots \end{align}
The first term of the fifth line of~\eqref{eq:expansion} combines several two-loop diagrams, given by
\begin{align} \sunrisetwolegs &= \frac12 \alpha^{2} g^{4} \frac{1}{3-D}\,\frac{2}{2-\frac{D}{2}}\,, \nonumber\\
\bubbletadpolelegsright = \bubbletadpolelegsleft &= \frac14 \alpha^{2} g^{4} \frac{1}{1-\frac{D}{2}}\,\frac{2}{2-\frac{D}{2}}\,. \end{align}
The second term is a three-loop tadpole,
\begin{align} \vertexthreelooptwo = \frac{1}{48} \alpha g^{6} \frac{1}{\left(1-\frac{D}{2}\right)^{3}} \, ,\quad \ldots \end{align}
Further, the last line of~\eqref{eq:expansion} comes from vacuum diagrams. 
When comparing these computations with the expansion~\eqref{eq:expansion}, note
that taking the $\partial_\varphi^n$ derivative to extract the contribution with $n$ legs
produces an additional factor of $n!$.
We invite the reader to confirm some of the remaining terms in the expansion using the tropical Feynman rules on their own.

The contribution of the tadpole diagrams is divergent when $D\rightarrow 2$. This 
is an expected feature of the tropical QFT: it has the same divergences as the original QFT. Notice also that the differential operator~\eqref{eq:PD} heavily degenerates when $D =2$. The explicit singularities in the perturbative expansion are manifestations of this degeneration.

\subsection{Diagrammatic renormalization of tropical sine-Gordon QFT}
\label{sec:renormalization}
In this section, we will make heavy use of dimensional regularization using the spacetime parameter $D$.
We use the notation $\varepsilon = 2-D$, typical for theories that are critical in $D=2$.

Because sine-Gordon QFT is renormalizable, all (UV) divergences can be absorbed by a change of the coupling parameters. 
For the sine-Gordon model specifically, only the coupling $\alpha$ needs to be renormalized~\cite{Coleman1975} to render the $n$-point functions finite. 
Coleman~\cite{Coleman1975} gave the diagrammatic explanation for this: all divergences come from (potentially) multiple tadpole insertions at some vertex. The series of such isolated tadpole insertions starts as 
\begin{equation}
\label{eq:vertexren-definition}
\vertextwo +       \vertextwoloop + \vertextwolooptwo +\ldots = \vertextworen \, ,
\end{equation}
where we denote a vertex that is decorated with all possible tadpole insertions as a hollow vertex. Note
that these tadpole insertions can appear at any vertex, not only at the 2-point vertex.
Dressing vertices with tadpole graphs is a classic staple of perturbative treatments 
of the sine-Gordon model, see \cite[Figure~4]{AmitGoldschmidtGrinstein1980}.

For instance, the contribution of the vacuum tadpole correction is given by the product of the vertex factor, symmetry factor and the Hepp bound,
\begin{equation}
\label{eq:tadpoleexample}
    \vertexloop
    =
    -
    \alpha \, \frac{1}{2}
\cdot
    \frac{1}{1-D/2}
    =    
-\frac{\alpha}{\varepsilon} \, ,
\end{equation}
where $\varepsilon = 2-D$. Due to the factorization rule for tropical Feynman integrals (see \S\ref{sec:perturbation}), we have 
\begin{equation}
\vertextwolooptwo = - \alpha g^4 \frac{1}{2! \cdot 2^2}
\cdot
    \frac{1}{(1-D/2)^2} 
= 
-
\frac{1}{2!}
\frac{\alpha g^4}{\varepsilon^2} \, ,
\end{equation}
where the analytic piece of the Feynman integral is the square of the analytic piece of~\eqref{eq:tadpoleexample}.
In general, the symmetry factor of $n$ tadpole insertions is ${1}/{(n! 2^n)}$,
the tropical Feynman integral is $\mathcal H_D(C_1)^n = \left(\frac{2}{\varepsilon} \right)^n$,
and if the tadpole has $2k$ legs,
then the coupling factor is $-\alpha (-g^2)^{k+n-1}$. Hence, we can resum the series for
the dressed two-point vertex,
\begin{equation}
\label{eq:vertexren}
 \vertextworen = 
\vertextwo +    \vertextwoloop + \vertextwolooptwo +\ldots  = 
-\alpha +
\frac{\alpha g^2}{\varepsilon} 
- 
\frac{1}{2!}
\frac{\alpha g^4}{\varepsilon^2} + \ldots  =
-
\alpha \cdot
e^{-\frac{g^2}{\varepsilon}} \, .
\end{equation}
The same daisy summation applied to a $2n$-vertex gives the renormalized $2n$-vertex with the same multiplicative
factor $e^{-g^2/\varepsilon}$, so the whole cosine renormalizes
multiplicatively.\footnote{This is the tropical realization of the normal-ordering identity~$\cos(g\varphi) = e^{-(g^2/2)G_0(0)}\!:\!\cos(g\varphi)\!:\!$,
where $G_0(x-y)=\langle\varphi(x)\varphi(y)\rangle_0$ is the regulated free-field correlator used to define normal ordering. Its tropical coincident-point contraction
is the one-loop tadpole integral, $G_{0,\mathrm{tr}}(0)=\mathcal H_{2-\varepsilon}(C_1)=2/\varepsilon$.
Exponentiating this contraction sums any number of independent tadpoles attached to the same vertex.}

This motivates us to define
\begin{align} \label{eq:renormalpha} \alpha_{\ren} = \alpha e^{-\frac{g^2}{\varepsilon}} \, , \end{align}
such that $ \vertextworen$ gets the \emph{renormalized} Feynman rules 
\begin{align} \vertextworen = -\alpha_\ren, \quad \vertexfourren = \alpha_\ren g^2, \quad\ldots \, . \end{align}
This definition is analogous to Coleman's renormalization of the sine-Gordon model. 
As we will see below,  the renormalization~\eqref{eq:renormalpha} is sufficient to render the tropical theory finite, just as for the original, undeformed theory.

\subsection{The renormalized tropical loop equation}
\label{subsec:renotroploop}
We can substitute the renormalized mass $\alpha_\ren$ into the effective action $\Gamma^\tr$. 
Expressing $\alpha$ in terms of $\alpha_\ren$ via $\alpha = \alpha_\ren e^{g^2/\varepsilon}$
gives the PDE
\begin{align} \label{eq:TrLoopEqRen} \mathcal P_{D}^\ren \Gamma^{\mathrm{tr}} = \bigl(1-\partial_\varphi^2\Gamma^\tr\bigr)^{-1}-1 \, , \end{align}
where $\mathcal P_{D}^\ren$ is the operator $\mathcal P_{D}$ from~\eqref{eq:PDsG} expressed in terms of the renormalized mass $\alpha_\ren$:
\begin{align} \label{eq:PDeps} \mathcal P_{D}^\ren &= -2+\frac12 \varepsilon (2 - \varphi\,\partial_\varphi + g \partial_{g} ) + (2- g^2) \, \alpha_\ren \partial_{\alpha_\ren} \,. \end{align}
The precise form of this operator follows from the chain rule as
$\mathcal P_D = -D - (1 - D/2)\varphi\partial_\varphi + (1-D/2)g\partial_g + 2\alpha\partial_\alpha$, which gives
$\alpha\partial_\alpha = \alpha_\ren\partial_{\alpha_\ren}$ and
$g\partial_g\big|_{\alpha} = g\partial_g\big|_{\alpha_\ren} - \tfrac{2g^2}{\varepsilon}\alpha_\ren\partial_{\alpha_\ren}$.
Substituting and using $D = 2 - \varepsilon$ gives \eqref{eq:PDeps}.

The expansion~\eqref{eq:expansion} rewritten in terms of $\alpha_\ren$ gives a perturbative solution to this new PDE, 
\begin{align} \label{eq:expansion_ren} \Gamma^{\mathrm{tr}} &= \Lambda + \frac{\alpha_\ren}{g^2} - \frac{\alpha_{\ren}}{g^2} \left(\frac{g^2 \varphi^2}{2} - \frac{g^4 \varphi^4}{4!} + \ldots \right) \nonumber \\
&+ \alpha^2_\ren \left(\frac{1}{2} - \frac{g^2 \varphi^2}{2} + \frac{g^4 \varphi^4}{6} + \ldots\right) - \alpha^3_\ren \left(\frac{1}{4} - \frac{3g^2 \varphi^2}{8} + \frac{7g^4 \varphi^4}{32} + \ldots\right) + O(\varepsilon) + \ldots \nonumber \\
&+ g^2 \alpha^2_\ren \left(\frac{g^2 \varphi^2}{2} - \frac{g^4 \varphi^4}{6} + \ldots \right) - g^2 \alpha^3_\ren \left(- \frac{1}{2} + \frac{13 g^2 \varphi^2}{8} - \frac{7g^4 \varphi^4}{6} + \ldots \right) + O(\varepsilon) + \ldots \nonumber \\
& + \ldots \, , \end{align}
with a $\varphi$-independent ``vacuum'' constant, which is infinite in the $\varepsilon \rightarrow 0$ limit,
\begin{equation}
    \Lambda = -\frac{\alpha}{g^2} =-\frac{\alpha_\ren}{g^2} e^{g^2/\varepsilon}  \, .
\end{equation}
We explicitly see that  for finite $\varepsilon$, in the classical limit $\hbar \rightarrow 0$, as we take it in Eq.~\eqref{eq:Boundary_Cond}, the coupling $\alpha$ coincides with $\alpha_\ren$ and the first line of \eqref{eq:expansion} does indeed reproduce the boundary value required in Eq.~\eqref{eq:Boundary_Cond} while the rest of the expansion is trivially zero in this limit. 

Apart from the constant $\Lambda$ the renormalized action is finite for $\varepsilon \rightarrow 0$, as expected by analogy to non-tropicalized sine-Gordon theory.  To make everything finite,  we will subtract this constant $\Lambda$, which does not affect the correlators of the field $\varphi$. Subsequently, we can  finish the renormalization procedure by taking the limit $\varepsilon \rightarrow 0$ while keeping the renormalized coupling $\alpha_\ren$ fixed. This way, we obtain a finite generating functional, which we define by
\begin{equation}
\label{eq:gamma_ren}
    \Gamma^{\text{tr}}_\ren(\varphi,\alpha_\ren,g) = \lim_{\varepsilon \rightarrow 0}  \left( \Gamma^{\text{tr}}  - \Lambda \right) \, .
\end{equation}
The fully renormalized effective action $\Gamma^{\text{tr}}_\ren$  satisfies the PDE,
\begin{equation}
\label{eq:PDEren}
    \mathcal{P}^\ren  \Gamma^{\text{tr}}_\ren = \bigl(1-\partial_\varphi^2\Gamma^\tr_\ren\bigr)^{-1}-1 \, ,
\end{equation}
where we can now take the (non-singular) limit $\varepsilon \rightarrow 0$ to obtain the operator $\mathcal{P}^\ren$ defined by
\begin{equation}
\label{eq:Pren}
    \mathcal{P}^\ren = \lim_{\varepsilon \rightarrow 0} \mathcal{P}^\ren_{D}= -2 +(2- g^2) \, \alpha_\ren \partial_{\alpha_\ren} \, .
\end{equation}
The boundary condition (i.e.~the $\hbar \rightarrow 0^+$ limit) becomes
\begin{equation}
\label{eq:Boundary_Cond_ren}
    \lim_{\hbar \rightarrow0} \hbar \, \Gamma^{\text{tr}}_\ren(\hbar^{-\frac{1}{2}} \varphi, \alpha_\ren, \hbar^{\frac{1}{2}} g)= \frac{\alpha_\ren}{g^2} \cos{g\varphi}  \, ,
\end{equation}
with no mention of the bare coupling $\alpha$ and no $\varepsilon$'s. 

We briefly remark that 
Eq.~\eqref{eq:PDEren} can also be explained purely diagrammatically along the lines of the proof of~\eqref{eq:Tr_Loop_Eq} from \cite{Borinsky:2025stz}.
The derivation of \eqref{eq:Tr_Loop_Eq} from \cite{Borinsky:2025stz} uses the recursive identity~\eqref{eq:hepp_recursion} and
    combines it with combinatorially gluing a string of 1PI graphs on a necklace. 
Eq.~\eqref{eq:PDEren} can be derived similarly,
by absorbing the singular tadpole subgraphs into the renormalized vertex
and dressing each edge with arbitrary 1PI two-point tadpole insertions. The right-hand side of \eqref{eq:PDEren} retains the same geometric-series form as the right-hand side of \eqref{eq:Tr_Loop_Eq}: combinatorially, it still builds a necklace of 1PI graphs. 

\subsection{Ambiguities of the renormalized effective action}
\label{sec:ambiguities}
Note that the right-hand side of \eqref{eq:Boundary_Cond_ren} includes a constant term (once $\cos{g \varphi}$ is expanded) in contrast to the original boundary condition \eqref{eq:Boundary_Cond}.
In fact, Eq.~\eqref{eq:PDEren} does not completely fix $\Gamma^\tr_\ren$. Specific, scheme-dependent ambiguities remain. The perturbative solution~\eqref{eq:expansion_ren} that we constructed from the (unique) bare solution $\Gamma^\tr$ in $D$ dimensions corresponds to
a specific choice of a renormalization scheme: the \emph{minimal subtraction scheme}, which is determined
by the most economical removal of all $\frac{1}{\varepsilon}$ poles.

We could also have subtracted additional finite terms in the definition of $\Gamma^\tr_\ren$~\eqref{eq:gamma_ren}. Associated with this is a shift ambiguity of the effective action 
\begin{align} \widetilde \Gamma^\tr_\ren = \Gamma^{\text{tr}}_\ren - \frac{C(g^2)}{g^2} \cdot \alpha_\ren^{\frac{2}{2-g^2}} \, , \end{align}
where $C(g^2)$ is some function in $g^2$ such that $\lim_{g\rightarrow 0} C(g^2) < \infty$. The shifted $\widetilde \Gamma^\tr_\ren$ solves~\eqref{eq:PDEren} just as $\Gamma^\tr_\ren$. Only the first coefficient of the $g^2$ expansion of $C(g^2)$ modifies the boundary condition \eqref{eq:Boundary_Cond_ren}. So, if $C(0)=0$ the boundary condition is also left intact. In fact, with $C(0)=1$, 
we recover the original boundary condition
$ \lim_{\hbar \rightarrow0} \hbar \, \widetilde\Gamma^{\text{tr}}_\ren(\hbar^{-\frac{1}{2}} \varphi, \alpha_\ren, \hbar^{\frac{1}{2}} g)= -\frac{\alpha_\ren}{g^2} (1-\cos{g\varphi})$.

A second ambiguity comes from the freedom to multiplicatively change 
$\alpha_\ren$. Both~\eqref{eq:PDEren} and~\eqref{eq:Boundary_Cond_ren} are invariant
under the change $\alpha_\ren \rightarrow A(g^2) \cdot \alpha_\ren$, whenever $A(g^2) = 1 + \mathcal O(g^2)$.
We can directly associate this second ambiguity with the renormalization group which, here, acts multiplicatively on $\alpha_\ren$. In the next subsection, we will illuminate this from the perspective of the $\beta$ function associated with $\alpha_\ren$.

It is convenient to fix $\Gamma^\tr_\ren$ to agree with the minimal subtraction scheme solution without reference to the original bare solution by imposing the additional constraint on top of~\eqref{eq:PDEren} and~\eqref{eq:Boundary_Cond_ren} that
\begin{align} \label{eq:minimalsubstraction} \Gamma^\tr_\ren = \frac{\alpha_\ren}{g^2} \cos ( {g\varphi}) + \mathcal O(\alpha_\ren^2) \, , \end{align}
and $\Gamma^\tr_\ren \in \Q(g^2)[\cos(g \varphi)][[\alpha_\ren]]$, which 
means that $\Gamma^\tr_\ren$ is a power series in $\alpha_\ren$ with coefficients that are polynomials in $\cos(g \varphi)$ and rational functions in $g^2$.
Ansatz~\eqref{eq:ansatz}, which we use later on for solving the loop equation, clearly fulfills this condition.

\subsection{The \texorpdfstring{$\beta$}{beta}-functions of the tropical sine-Gordon QFT}
\label{sec:betafunctions}

Eq.~\eqref{eq:Pren} gives us the $\beta$-function of tropical sine-Gordon QFT in 
the minimal subtraction scheme,\footnote{Here, we get a minus sign via the typical $\beta_X = \mu \partial_\mu X$ definition of the $\beta$-function.}
\begin{align} \label{eq:betaalpha} \beta_{\alpha_\ren} = - (2-g^2) \alpha_\ren \, . \end{align}
Moreover, we find that the $\beta$-function associated with the coupling $g$ vanishes identically, $\beta_{g^2} = 0$. So, in the tropical QFT, the $g^2$-coupling does not run at all. We emphasize that our $\beta$ function computations are \emph{exact} in the minimal subtraction scheme and not subject to any truncation. Interestingly, the tropical $\beta$-functions agree with those of the original, non-tropicalized sine-Gordon QFT up to first order in the coupling parameter $\alpha$, up to trivial rescalings.
From~\cite{AmitGoldschmidtGrinstein1980}, 
\begin{align} \beta_{\alpha_\ren}^{sG} &= - \left(2- \frac{g^2}{4 \pi} \right) \alpha_\ren + \mathcal O(\alpha_\ren^2) \, ,\\
\beta_{g^2}^{sG} &= 0 + \mathcal O(\alpha_\ren^2) \, , \end{align}
we see that the tropical model perfectly reproduces the first coefficients of the $\beta$ functions of the original model up to the redefinition $g^2 \rightarrow \frac{g^2}{4\pi}$. 
These are precisely the renormalization-scheme-independent coefficients of the sine-Gordon $\beta$ function.
The non-tropical $\beta$-function changes sign at the phase-transition point $g^2 = 8 \pi$, where the theory ceases to have a ground state. The analogue of this phase-transition point in the tropical model is therefore the value $g^2=2$, where the $\cos$ in the sine-Gordon potential becomes irrelevant and the theory becomes Gaussian. 

\subsection{Resumming the perturbative expansion}
\label{sec:resumming}
The renormalized expansion~\eqref{eq:expansion_ren} features 
three expansion parameters: $\alpha_\ren$, $g^2$ and $\varphi$. In this section, we briefly illustrate various possible resummations of the expansion. 
It is instructive to observe how these resummations are compatible with both the PDE viewpoint on tropical QFT 
and the traditional diagrammatics.

The diagrammatic explanation for the infinite number of terms at each order in $g^2$ is the presence of renormalized two-valent vertices $\vertextworen$. We start by resumming these terms systematically, guided by the diagrammatics.

The tadpole $C_1$ has already been absorbed into the renormalized vertex $\alpha_\ren$. On top of that tadpole, we can also resum all other tadpole one-loop contributions to the two-point function. Starting with the renormalized 2-point vertex and then the cycle $C_2$ with one renormalized 2-point vertex inserted, we have the diagrammatic expansion up to one loop,
\begin{equation}
\begin{gathered}
\label{eq:twopointGammaexp}
\left . \frac{\partial^2 \Gamma^\tr_\ren}{\partial \varphi^2} \right|_{\varphi = 0}
=
\quad
 \vertextworen
\quad
+
\quad
\begin{tikzpicture}[baseline=-0.5ex,scale=0.85,line cap=round,line join=round] \coordinate (L) at (-0.45,0); \coordinate (R) at (0.45,0); \draw[thick] (L) to[bend left=35] (R); \draw[thick] (L) to[bend right=35] (R); \draw (L)--++(160:0.45); \draw (L)--++(200:0.45); \foreach \v in {L,R} \filldraw[fill=white,draw=black,thick] (\v) circle (0.07); \end{tikzpicture}
\quad+\quad
\begin{tikzpicture}[baseline=-0.5ex,scale=0.75,line cap=round,line join=round] \coordinate (v1) at (-0.75,0); \coordinate (v2) at (0.45,0.55); \coordinate (v3) at (0.45,-0.55); \draw[thick] (v1)--(v2)--(v3)--cycle; \draw (v1)--++(160:0.55); \draw (v1)--++(200:0.55); \foreach \v in {v1,v2,v3} \filldraw[fill=white,draw=black,thick] (\v) circle (0.07); \end{tikzpicture}
\quad+\quad\cdots \, .
\end{gathered}
\end{equation}
The tropical Feynman integral of $C_n$ is computed in Eq.~\eqref{eq:HeppCycle}. 
The symmetry factor of each term in the expansion is $\frac12$ due to the horizontal flip symmetry.
Thus we get the expansion
\begin{align} \label{eq:Gamma2_oneloop_cycle_corrected} \left. \partial_\varphi^2\Gamma^\tr_\ren \right|_{\varphi=0}&= -\alpha_\ren + \frac12 \alpha_\ren g^2 \sum_{r \geq 2} \mathcal H_2( C_r ) (-\alpha_\ren)^{r-1} +\mathcal O(g^4)\\
&= -\alpha_\ren + \frac12 \alpha_\ren g^2 \sum_{r \geq 2} \frac{r}{r-1} (-\alpha_\ren)^{r-1} +\mathcal O(g^4)\\
&= -\alpha_\ren + \frac12 \alpha_\ren g^2 \sum_{r \geq 1} \frac{r+1}{r} (-\alpha_\ren)^{r} +\mathcal O(g^4)\\
&= -\alpha_\ren - \frac12 \alpha_\ren g^2 \left( \log( 1 + \alpha_\ren) + \frac{\alpha_\ren}{1+\alpha_\ren} \right) +\mathcal O(g^4) \, . \end{align}
Here, we can see explicitly how $\log$-term singularities in $\alpha_\ren$ emerge from the resummation. 

A different resummation combines the $\varphi$ terms into factors of $\cos(g\varphi)$.
The necessary observation is that we may decorate the renormalized vertices in~\eqref{eq:twopointGammaexp} with arbitrary legs. We mark each occurrence of a leg-decorated vertex with a $\partial_\varphi^2 (1-\frac{\alpha_\ren}{g^2} \cos (g \varphi)) = \alpha_\ren \cos (g \varphi) $
term. The second-order derivative fixes two legs of the renormalized vertex that are tied 
together to form the necklace. The effective action is, hence, naturally organized as an expansion in
$h = \alpha_\ren \cos(g\varphi)$ completely absorbing the $\varphi$ dependence. 
With this parametrization, Eq.~\eqref{eq:PDEren} transforms into 
\begin{align} \label{eq:PDEren_h} (-2 +(2- g^2) \, ( \alpha_\ren \partial_{\alpha_\ren} + h \partial_h )) \Gamma^{\text{tr}}_\ren = \bigl(1+g^2 ( h \partial_h - ( \alpha_\ren^2 - h^2) \partial_h^2 )\Gamma^\tr_\ren\bigr)^{-1}-1 \, , \end{align}
and the boundary condition becomes
\begin{align} \lim_{\hbar \rightarrow0} \hbar \, \Gamma^{\text{tr}}_\ren(h, \alpha_\ren, \hbar^{\frac{1}{2}} g)= \frac{h}{g^2} \, . \end{align}
Solving \eqref{eq:PDEren_h} via an iteration method as an expansion in $g^2$ then gives through two loops,
\begin{align} \label{eq:GammaLoopExpansion} \Gamma^{\text{tr}}_\ren = \frac{h}{g^2} + \frac{h}{2} \log{(1+h)} \, + \, & g^2 \Biggl(\frac{h}{8} \log^2{\left(1+h \right)} + \frac{h^2}{4(1+h)} \log{\left(1+h \right)} \nonumber \\
&- \frac{2 h^2\left(6 + 6 h + h^2\right) - \alpha^2_\ren \left(12 + 15 h + 5 h^2 \right)}{24(1+h)^3} \Biggr) + \ldots \, . \end{align}
The leading logarithms (LL) can be resummed to the interesting formula
\begin{equation}
    \left(\Gamma^{\text{tr}}_\ren \right)^{\text{LL}} = \frac{h}{g^2} (1+h)^{g^2/2} \, .
\end{equation}

Alternatively, the renormalized tropical loop equation may also be expanded order-by-order in $\alpha_\ren$, keeping all orders in $g^2$ as $g$ only enters as a free parameter. E.g., up to order $\alpha_\ren^2$,
\begin{equation}
\label{eq:alpha2_ren}
    \Gamma^{\text{tr}}_\ren
    =  \frac{\alpha_\ren}{g^2} \cos{g \varphi}
    + \frac{\alpha_\ren^2}{2(1+g^2)} \left( \frac{g^2}{1-g^2} + \cos^2{g \varphi} \right)
    +\ldots \, .
\end{equation}
\subsection{Solving the renormalized tropical sine-Gordon PDE}
\label{sec:solvingpde}

Eq.~\eqref{eq:alpha2_ren} suggests making the following Ansatz for $ \Gamma^{\text{tr}}_\ren$
\begin{align} \label{eq:ansatz} \Gamma^{\text{tr}}_\ren = \sum_{k = 1}^\infty \sum_{j = 0}^k c_{k,j}(g) \alpha_\ren^k \cos(j g \varphi)\, . \end{align}
This Ansatz was successfully used for functional renormalization group analyses \cite{NagyPolonyiSailer2006,Nagy2007,Nagy:2009pj}.
The key property of this Ansatz is that $\alpha_\ren^k \cos(j g \varphi)$ is an eigenvector of both the operator $\mathcal P^\ren = -2 +(2- g^2) \, \alpha_\ren \partial_{\alpha_\ren}$ and $\partial_\varphi^2$:
\begin{align} \mathcal P^\ren \alpha_\ren^k \cos(j g \varphi) &= ( (2- g^2)k - 2 ) \alpha_\ren^k \cos(j g \varphi) \\
\partial_\varphi^2 \alpha_\ren^k \cos(j g \varphi) &= - j^2 g^2 \alpha_\ren^k \cos(j g \varphi) \,. \end{align}
The minimal subtraction renormalization condition~\eqref{eq:minimalsubstraction} implies that 
$c_{1,1} = \frac{1}{g^2}$ and $c_{1,0} = 0$.
Multiplying Eq.~\eqref{eq:PDEren} by $1-\partial_\varphi^2\Gamma^\tr_\ren$, expanding the Ansatz and using the identity $\cos(n x) \cos(m x) = \frac12 \left( \cos((n-m)x) + \cos( (n+m) x) \right)$, we get
\begin{align} \label{eq:recursion} B_{k,j} \, c_{k,j} = -\frac{g^2}{2} \sum_{\substack{k_1, k_2 \geq 1\\ k_1+k_2 = k }} \left( (2-g^2)\, k_1 - 2 \right) \sum_{j_1 = 0}^{k_1} \sum_{j_2 = 0}^{k_2} j_2^2 \, c_{k_1,j_1} c_{k_2,j_2} \left( \delta_{j_1+j_2,\, j} + \delta_{|j_1-j_2|,\, j} \right) \, , \end{align}
for all $j \leq k$ with $k \geq 2$ and 
where $\delta$ denotes the Kronecker delta and
    $B_{k,j} = (2-g^2)\,k - 2 + j^2 g^2\,.$
The right-hand side of \eqref{eq:recursion} only involves coefficients of the orders $k_1, k_2 < k$. Whenever $B_{k,j} \neq 0$, Eq.~\eqref{eq:recursion} hence determines the coefficients of order $k$ from the lower orders. Moreover, since $j_1 \leq k_1$ and $j_2 \leq k_2$, the right-hand side vanishes for $j > k$, which
shows that the Ansatz~\eqref{eq:ansatz} closes.
In the range $0 < g^2 < 2$, the condition $B_{k,j} \neq 0$ can only fail when 
$ g^2 = 2 - \frac{2}{k} \in \left\{1, \tfrac43, \tfrac32, \tfrac85, \ldots \right\} $. We can also allow these countably many values of $g^2$ at the cost of augmenting the 
Ansatz~\eqref{eq:ansatz} with $\log$ terms. It would be interesting to investigate 
if these \emph{resonances} in the tropical theory can be associated with 
the \emph{breather solutions} in the non-tropical sine-Gordon QFT.
For the present paper, we assume the value of $g^2$ not to fall into this 
countable set of resonant values.

\begin{figure}[p]
\centering
\begin{subfigure}[t]{0.41\linewidth}
\centering
\includegraphics[width=\linewidth]{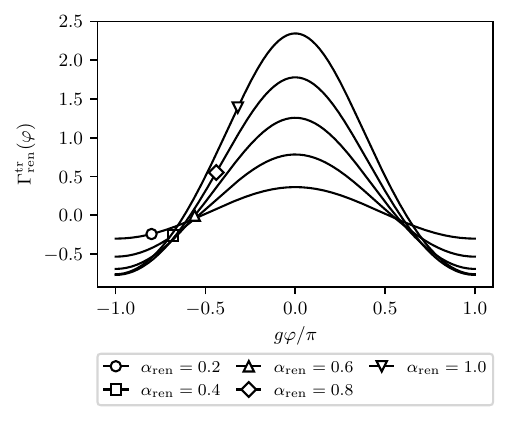}
\caption{$g^2 = 0.6$}
\label{fig:profiles_0p6}
\end{subfigure}\hfill
\begin{subfigure}[t]{0.41\linewidth}
\centering
\includegraphics[width=\linewidth]{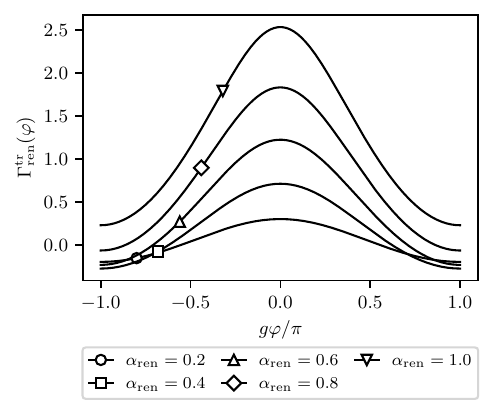}
\caption{$g^2 = 0.8$}
\label{fig:profiles_0p8}
\end{subfigure}

\begin{subfigure}[t]{0.41\linewidth}
\centering
\includegraphics[width=\linewidth]{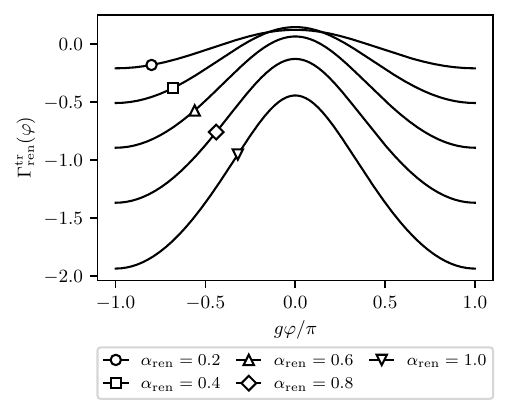}
\caption{$g^2 = 1.2$}
\label{fig:profiles_1p2}
\end{subfigure} \hfill
\begin{subfigure}[t]{0.41\linewidth}
\centering
\includegraphics[width=\linewidth]{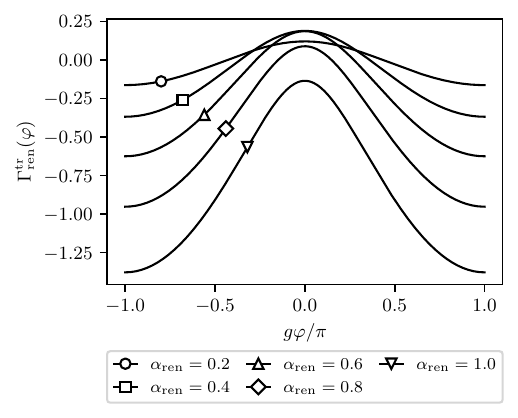}
\caption{$g^2 = 1.4$}
\label{fig:profiles_1p4}
\end{subfigure}
\caption{Plots of the effective action $\Gamma^{\tr}_\ren(\varphi)$
at couplings $g^2$ and $\alpha_\ren$ obtained using the 
convergent $\alpha_\ren$ expansion~\eqref{eq:ansatz}.
}
\label{fig:profiles}
\end{figure}

It turns out that the expansion~\eqref{eq:ansatz} 
has a \emph{non-zero} radius of convergence in $\alpha_\ren$
provided that $g^2$ is in the range $0 < g^2 < 2$.
This is interesting, as perturbative expansions of QFT quantities such as the effective 
action are typically not convergent. These expansions are usually
asymptotic in the sense that the coefficients grow factorially or even faster \cite{Dyson:1952tj,Lipatov1977}.
 Note that~\eqref{eq:ansatz} is
a systematic expansion in $\alpha_\ren$, but is exact in $g$: no power series expansion in $g^2$ is needed. In fact, expanding the individual coefficients $c_{k,j}(g)$ in $g^2$ and summing over all $k$ will result in a factorially divergent power series in $g^2$. We will analyse such a factorial divergence in the next section.

Because such a convergent perturbative expansion is an interesting find,  
we give an explicit and mathematically rigorous proof of convergence. 
The proof is surprisingly technical, and crucially relies on the coupling $g^2$ to lie in the 
``physical'' range $0 < g^2 < 2$.
\begin{theorem}
\label{thm:convergence}
Let $D(g) = \max\left(\frac{1}{2-g^2}, \frac43 \frac{1}{g^2} \right)$.
If $0 < g^2 < 2$ such that $g^2 \neq 2 - \frac{2}{k}$ for all $k \geq 1$,
 and $|\alpha_\ren| < \frac{1}{4D(g) (1+g^2D(g))}$, %
then, using the initial condition 
$c_{1,1} = \frac{1}{g^2}$, $c_{1,0} = 0$,
 we can solve~\eqref{eq:recursion} recursively for the coefficients $c_{k,j}$ and the sum in Eq.~\eqref{eq:ansatz} over $k$ and $j$ 
is absolutely convergent.
\end{theorem}
The assumption $g^2 \neq 2 - \frac{2}{k}$ for all $k \geq 1$ can be weakened to requiring the inequality only for all even $k \geq 2$, because
the potential poles in the recursion stemming from odd $k$ are canceled. 
\begin{proof}
As explained above, it follows from the assumptions that $B_{k,j} \neq 0$ for all $0 \leq j \leq k$ and $k \geq 2$ that the solution of the recursion~\eqref{eq:recursion} exists uniquely.

For all $k \geq 3/(2-g^2)$, 
we have $B_{k,j}= (2-g^2)k -2 + j^2 g^2 \geq 1$.
Let $h_1 = 1/g^2$ and $h_{k} = \sum_{j =0}^k |B_{k,j} c_{k,j}|$ for $k \geq 2$.
In particular, we have for $k \geq 3/(2-g^2)$ 
\begin{align} \sum_{j=0}^k |c_{k,j} \alpha_\ren^k \cos(j g \varphi) | \leq |\alpha_\ren^k| \sum_{j=0}^k \frac{1}{|B_{k,j}| } |B_{k,j} c_{k,j}| \leq |\alpha_\ren^k| h_k \, . \end{align}
Hence, we need to show that $h_k$ grows at most as $(4D(g) (1+g^2 D(g)))^k$ for $k\rightarrow \infty$ to prove the statement.

If $1 \leq j \leq k$ with $k \geq 2$, we have $B_{k,j} \geq (2-g^2) (k-1) + (j^2-1)g^2 \geq 2-g^2 + (j^2-1)g^2$. So,
\begin{align} \frac{j^2}{B_{k,j}} \leq \frac{1}{ \frac{1}{j^2} (2-g^2) + (1-\frac{1}{j^2})g^2} \leq \max\left(\frac{1}{2-g^2}, \frac43 \frac{1}{g^2} \right) = D(g)\,, \end{align}
where the first term in the maximum bounds the $j=1$ case and the second term the $j \geq 2$ case.

Using the above inequality and $c_{1,1} = h_1$, we get the following estimate
\begin{gather} \label{eq:ineq1} \sum_{j = 1}^k j^2 |c_{k,j}| \leq \sum_{j = 1}^k \frac{j^2}{B_{k,j}} |B_{k,j} c_{k,j}| \leq D(g) h_k \, , \end{gather}
for all $k \geq 1$, where we also used $D(g) \geq 1$.
Similarly, we have the estimate
\begin{gather} \label{eq:ineq2} \sum_{j = 0}^k |((2-g^2)k -2) c_{k,j}| = \sum_{j = 0}^k | (B_{k,j} - j^2 g^2) c_{k,j}| \leq h_k + g^2 D(g) h_k \,. \end{gather}
Summing Eq.~\eqref{eq:recursion} over $0\leq j \leq k$ and using 
$\sum_{j=0}^k \delta_{|j_1\pm j_2|,j} \leq 1$,
we find that for $k \geq 2$,
\begin{gather} h_k \leq \frac{1}{2} 2 g^2 \sum_{\substack{k_1, k_2 \geq 1\\ k_1+k_2 = k }} \sum_{j_1=0}^{k_1} \sum_{j_2=0}^{k_2} |((2-g^2)k_1 -2) c_{k_1,j_1}| \cdot |j_2^2 c_{k_2,j_2}| \, . \end{gather}
Using Eqs.~\eqref{eq:ineq1} and \eqref{eq:ineq2}, we get
\begin{gather} h_k \leq g^2 D(g)(1+ g^2 D(g)) \sum_{\substack{k_1, k_2 \geq 1\\ k_1+k_2 = k }} h_{k_1} h_{k_2} \, . \end{gather}
Let $\kappa(g) = g^2 D(g)(1+ g^2 D(g))$ and define 
$H_k = g^{-2k} \kappa(g)^{k-1} C_{k-1}$, where $C_k$ is the Catalan number which fulfills
the recurrence, $C_0=1$ and $C_{k-1} = \sum_{ \substack{k_1, k_2 \geq 1,k_1+k_2 = k } } C_{k_1-1} C_{k_2-1} $ for all $k \geq 2$. The recurrence
$ H_{k} = \kappa(g) \sum_{ \substack{k_1, k_2 \geq 1, k_1+k_2 = k } } H_{k_1} H_{k_2} $
follows and we get $h_k \leq H_k$ for $k\geq 1$.
Famously, the Catalan numbers behave as $C_n \sim 4^n/(\sqrt{\pi} n^{3/2})$ for $n\rightarrow \infty$.
\end{proof}

We emphasize that the proof only works in the range $0 < g^2 <2$. If 
$g^2 > 2$ or $g^2 < 0$, then $B_{k,j}$ becomes arbitrarily close to zero for large enough $j \leq k$. 
This implies that the resulting values $c_{k,j}$ have a dense set of poles in those ranges of $g^2$.
So, the Ansatz~\eqref{eq:ansatz} fails to produce a well-behaved solution in that regime.

Figure~\ref{fig:profiles} shows explicit plots of the effective 
action $\Gamma^\tr_\ren$ computed from the recursion~\eqref{eq:recursion} truncated at $k \leq 300$.
Of course, the plots are inherently dependent on our choice of the minimal subtraction scheme.
In Figure~\ref{fig:derivativedata}, we plot a scheme-independent quantity:
the value of the four-point 1PI vertex $\partial_\varphi^4 \Gamma^\tr_\ren|_{\varphi=0}$ expressed as
a function of $\partial_\varphi^2 \Gamma_\ren^\tr|_{\varphi=0}$ (instead of $\alpha_\ren$).

\subsection{Numerically solving the model at large coupling}
\label{sec:numerics}
Because the expansion~\eqref{eq:ansatz} gives us a sound and concrete description of $\Gamma^\tr_\ren$
for a finite range of values of $\alpha_\ren$, we can use these values as initial conditions for the
numerical integration of the PDE~\eqref{eq:PDEren}. 
To do so, we pass to a more general expansion of the form
\begin{align} \label{eq:ansatz_numerical} \Gamma^{\text{tr}}_\ren = \sum_{j = 0}^J c_{j}(\alpha_\ren, g) \cos(j g \varphi)\, , \end{align}
where we choose a cut-off $J \geq 0$ and ignore modes of higher order.
Plugging this Ansatz into the PDE~\eqref{eq:PDEren} and neglecting modes beyond the cut-off gives us a non-linear, but finite-dimensional system of ODEs that the coefficients $c_{j}(\alpha_\ren, g)$
need to fulfill and the convergent $\alpha_\ren$-expansion provides us with the necessary initial values for small, non-zero $\alpha_\ren$.

Provided that $g^2$ is not near the boundary of the interval $(0,2)$, where the solution can become singular (in fact, the solution only becomes singular when $g^2 \rightarrow 0$), we can numerically integrate this system to effectively arbitrarily large values of $\alpha_\ren$ using standard ODE 
integrators with good stability properties.

\begin{figure}[p]
\centering
\begin{subfigure}[t]{0.41\linewidth}
\centering
\includegraphics[width=\linewidth]{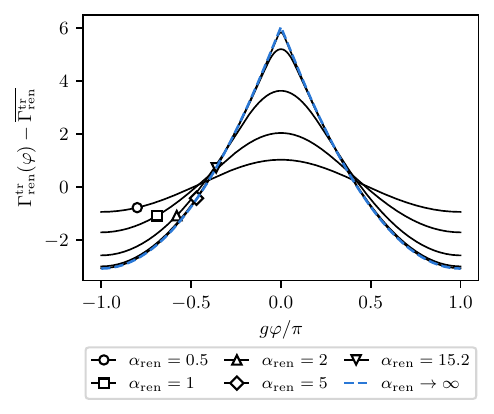}
\caption{$g^2 = 0.5$}
\label{fig:profiles_avg_0p5}
\end{subfigure}\hfill
\begin{subfigure}[t]{0.41\linewidth}
\centering
\includegraphics[width=\linewidth]{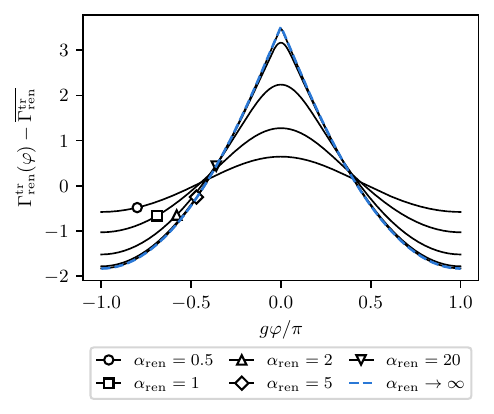}
\caption{$g^2 = 0.8$}
\label{fig:profiles_avg_0p8}
\end{subfigure}

\begin{subfigure}[t]{0.41\linewidth}
\centering
\includegraphics[width=\linewidth]{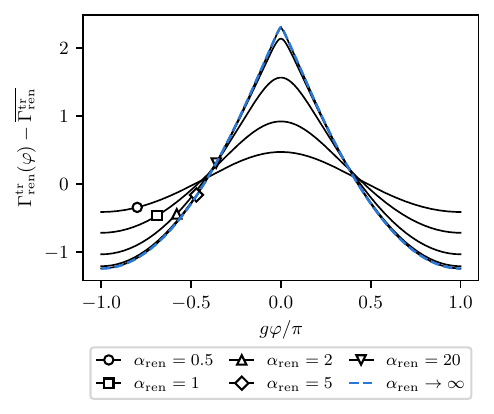}
\caption{$g^2 = 1.1$}
\label{fig:profiles_avg_1p1}
\end{subfigure}\hfill
\begin{subfigure}[t]{0.41\linewidth}
\centering
\includegraphics[width=\linewidth]{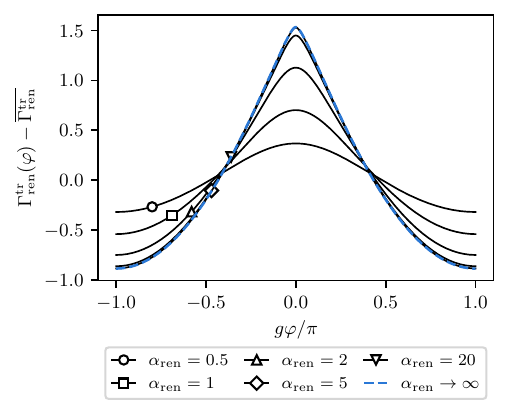}
\caption{$g^2 = 1.4$}
\label{fig:profiles_avg_1p4}
\end{subfigure}
\caption{Plots of the shifted effective action $\Gamma^{\tr}_\ren(\varphi) - \overline{\Gamma^{\text{tr}}_\ren}$
at couplings $g^2$ and $\alpha_\ren$ obtained using the 
numerical integration of the Ansatz~\eqref{eq:ansatz_numerical}. 
The blue dashed line depicts the scale-invariant solution computed using~\eqref{eq:newton}.
}
\label{fig:galerkin_profiles}
\end{figure}
Figure~\ref{fig:galerkin_profiles} shows the results of the numerical integration of the PDE~\eqref{eq:PDEren} using the Ansatz~\eqref{eq:ansatz_numerical} with $J = 320$ from the initial data of the recursion \eqref{eq:recursion}. We used 
SciPy's function \verb|solve_ivp| with \verb|method=Radau|  \cite{2020SciPy-NMeth} to perform this numerical integration. 
We subtracted the field-independent shift, $\overline{\Gamma^{\text{tr}}_\ren} = c_0(\alpha_\ren, g)$, which removes the zero-mode such that the approach of the solution becomes finite in the $\alpha_\ren \rightarrow \infty$ limit. This subtraction can also be understood as the inclusion of a vacuum counterterm. 

\subsection{Approaching the IR fixed point}
\label{sec:irfixedpoint}
The coupling $\alpha_\ren$ can be used to parameterize the renormalization group flow of the theory. From the expression for its $\beta$-function \eqref{eq:betaalpha}, one can see that we start at the Gaussian fixed point at $\alpha_\ren = 0$ and flow towards the IR fixed point when $\alpha_\ren \rightarrow \infty$.  At a renormalization group fixed point the theory becomes scale-invariant. 
Recall the discussion in Section~\ref{sec:trop_ren_group_fix}, where we showed that the scale-invariant solutions in $D=2$ solve Newton's equation~\eqref{eq:newton}. In 
this section, we will illustrate how in the limit $\alpha_\ren \rightarrow \infty$, after renormalization of vacuum contributions, our explicit solution of the renormalized tropical loop equation \eqref{eq:PDEren} approaches one of the scale-invariant solutions indexed by the parameter $E$.

The tropical theories, which we constructed using Eq.~\eqref{eq:ansatz}, have a periodic effective action with period $T = \frac{2\pi}{g}$.  For $g^2 \rightarrow 2$, we recover the limiting period $T = \sqrt{2} \pi$ from Section~\ref{sec:trop_ren_group_fix},
which is attained when the `energy' $E$ of the oscillation is small. As the period of the periodic scale-invariant solutions is a monotonically increasing function $T(E)$ and the coupling constant $g^2$ does not flow, we can invert this function and get the value of $E$ associated with every value of $g^2 \in (0,2)$:  $E(g) = T^{-1}( \frac{2\pi}{g})$.
Because $\frac{2\pi}{g} \rightarrow \infty$ as $g \rightarrow 0^+$, the range $g^2 \in (0,2)$ covers the whole range of positive values $E(g) \in (0,\infty)$.
Hence, the set of Lagrangians~\eqref{eq:S} of tropical sine-Gordon theory indexed by the values of $g^2\in (0,2)$ covers the whole set of periodic scale-invariant tropical QFTs in $D=2$ with the positivity property $E >0$. 
Each Lagrangian flows towards the respective fixed point in the IR.

We remark that in the UV we can interpret the theory as a real free scalar QFT with shift symmetry valued in $\R$. The potential $\alpha \cos(g\varphi)$ breaks this symmetry down to the discrete shifts by $\frac{2\pi n}{g}$, which form a subgroup $\frac{2\pi n}{g} \,\mathbb{Z} \subset \R$. The quantized theory retains the same symmetry. Hence, we have the expansion~\eqref{eq:ansatz} of the effective action which respects this residual, discrete symmetry.

\begin{figure}[t]
\centering
\includegraphics[width=0.74\linewidth]{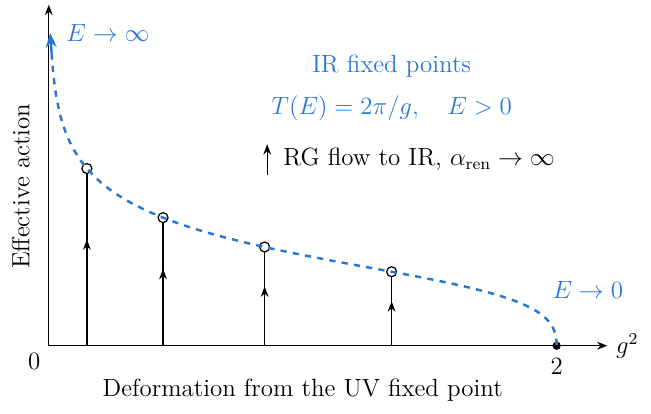}
\caption{Schematic view of the tropical sine-Gordon RG flow. The x-axis shows values of $0<g^2<2$ parameterizing the bare sine-Gordon Lagrangians. The y-axis schematically represents the effective actions at different scales, given by different values of $\alpha_\ren$. Black arrows schematically indicate the RG flow with increasing $\alpha_\ren$ at fixed $g^2$. The blue dashed curve is the limit $\alpha_\ren \rightarrow \infty$.}
\label{fig:traj}
\end{figure}
Figure~\ref{fig:galerkin_profiles} shows the periodic scale-invariant solution as a dashed blue line for different values of $g^2$. For each value of $g^2$, the respective IR fixed point is clearly approached in the $\alpha_\ren \rightarrow \infty$ limit. In Figure~\ref{fig:traj}, we illustrate the RG behaviour of the theory: starting from a free theory in the UV, one introduces a relevant cosine deformation parameterized by the coupling $g$ and flows down via the running coupling $\alpha_\ren$ to a stationary solution of \eqref{eq:newton} in the IR.

\subsection{Comparison to the Lukyanov--Zamolodchikov formula}
\label{subsec:comparison-sine-gordon}

In \cite{Lukyanov1996}, 
Lukyanov and Zamolodchikov conjectured the following formula (LZ formula) for the generating function of the field $\varphi$ evaluated at the origin in
sine-Gordon theory~\eqref{eq:S}, 
\begin{gather} \begin{gathered} \label{eq:LZ_one_point} \frac{1}{Z} \int \mathcal{D} [\varphi]e^{ia \varphi(0)} e^{-S_{sG}[\varphi,0] } = \left( \frac{ m\,\Gamma\!\left(\frac{1+\xi_{\mathrm{LZ}}}{2}\right) \Gamma\!\left(1-\frac{\xi_{\mathrm{LZ}}}{2}\right)} {4\sqrt{\pi}} \right)^{\frac{a^2}{4\pi}} \\
 \times \exp\!\left( \int_0^\infty \frac{\dd t}{t} \left( \frac{\sinh^2\!\left(\frac{agt}{4\pi}\right)} {2\sinh\!\left(\frac{g^2t}{8\pi}\right) \sinh t\cosh\!\left(\left(1-\frac{g^2}{8\pi}\right)t\right)} -\frac{a^2}{4\pi}e^{-2t} \right) \right) \,, \end{gathered} \end{gather}
where $a$ is a scalar parameter, $m$ is a specific sine-Gordon mass scale\footnote{See Eq.~(12) of \cite{Lukyanov1996} for details on the exact relationship with the Lagrangian mass in sine-Gordon theory.},
$\xi_{\text{LZ}}=g^2/(8\pi-g^2)$
and $Z$ is the normalization constant (Eq.~(20) of \cite{Lukyanov1996}).
The one-dimensional integral on the right-hand side can be evaluated 
using standard quadrature techniques. 

\begin{figure}[t]
\centering
\begin{subfigure}[t]{0.49\linewidth}
\centering
\includegraphics[width=\linewidth]{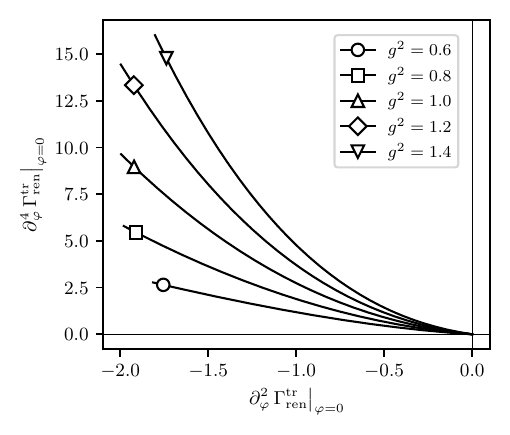}
\caption{Tropical sine-Gordon theory}
\label{fig:derivativedata_tr}
\end{subfigure}\hfill
\begin{subfigure}[t]{0.49\linewidth}
\centering
\includegraphics[width=\linewidth]{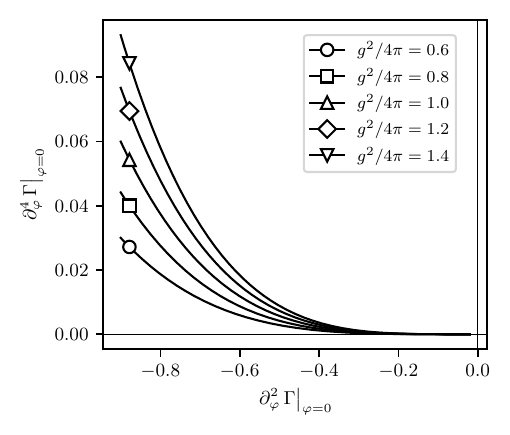}
\caption{Undeformed sine-Gordon from the LZ formula}
\label{fig:derivativedata_lz}
\end{subfigure}
\caption{Fourth versus second $\varphi$-derivative of the effective
action at $\varphi = 0$ evaluated at $x=0$ in undeformed sine-Gordon theory (recall that $\varphi$ is a function of the position $x$ before tropicalization). %
By construction, the graphs are scheme independent.}
\label{fig:derivativedata}
\end{figure}

Figure~\ref{fig:derivativedata} shows 
the value of $\frac{\partial^4 \Gamma}{\partial \varphi^4} $ versus  $\frac{\partial^2 \Gamma}{\partial \varphi^2}$. From the LZ formula~\eqref{eq:LZ_one_point}, we used the standard Legendre-transform procedure to extract the value of the effective action. The comparison of these two derivatives has the convenient property that it is renormalization group invariant:
we effectively remove the dependence on $\alpha_\ren$ and $m$ on the respective side.
The plots in Figure~\ref{fig:derivativedata} show the qualitatively similar behaviour of the traditional sine-Gordon model and
the tropicalized model. Note, however, the very different scales of the two plots.
We translated the coupling $g^2 \rightarrow g^2/ 4 \pi$ as suggested by the $\beta$-function dictionary above.

\subsection{Large-order behaviour of the \texorpdfstring{$g^2$}{g-squared}-perturbative expansion}
\label{sec:largeorder}

In the preceding sections, we discussed the solution of the PDE~\eqref{eq:PDEren} and thereby the solution 
of the renormalized tropical sine-Gordon model. We constructed the solution via a series expansion in the 
parameter $\alpha_\ren$, which has a non-zero radius of convergence. This is unusual for a perturbative computation in
QFT. 
To further illuminate this interesting feature, we study the expansion in $g^2$ in this section.

First,  we propose a suitable, slight variant of our original tropical model for such an investigation. The LZ formula~\eqref{eq:LZ_one_point} reveals that the dependence on the effective mass $m$, which corresponds to the coupling $\sqrt{\alpha}$, drops out of all the $n$-point functions except the two-point function. A power-counting argument explains why: in the undeformed sine-Gordon model every vertex brings a power of $\alpha$, every propagator an $\alpha^{-1}$, and every loop integration an $\alpha^{D/2}$. The latter two are determined by dimensional analysis alone, as $\alpha$ supplies the only scale in the momentum integrals. An amputated graph $G$, stripped of the overall momentum conservation, therefore carries the weight
\begin{align} \label{eq:alphapowercounting} \alpha^{|V_G|} \, \alpha^{-|E_G|} \, \alpha^{\frac{D}{2} L(G)} = \alpha^{|V_G|-|E_G|+L(G)} \, \alpha^{-\frac{\varepsilon}{2} L(G)} = \alpha^{1-\frac{\varepsilon}{2} L(G)} \, . \end{align}
The Euler characteristic formula gives $|V_G|-|E_G|+L(G)=1$. Thus at $D=2$, i.e. in the limit $\varepsilon = 2-D \rightarrow 0$, every graph carries just a single overall power of $\alpha$.  In contrast, in the tropical theory every propagator is equal to $1$, and the infinite series of the two-valent vertices does not resum into a propagator, as it would happen in an undeformed QFT. 
A way to resolve this issue is to remove these two-valent insertions in the boundary condition~(\ref{eq:Boundary_Cond}), and set $\alpha_\ren =1$ to avoid the extra weighting of the vertices that also carry a power of $\alpha_\ren$. 
We thus define $\widehat\Gamma^\tr_\ren$ to satisfy the same PDE~\eqref{eq:PDEren} as $\Gamma^\tr_\ren$, but with the new boundary condition with the constant and quadratic terms subtracted:
\begin{equation}
\label{eq:Boundary_Cond_ren_quadratic_subtracted}
    \lim_{\hbar \rightarrow0}
    \hbar \,
    \widehat\Gamma^{\text{tr}}_\ren
    (\hbar^{-\frac{1}{2}} \varphi, \alpha_\ren, \hbar^{\frac{1}{2}} g)
    =
    -\frac{\alpha_\ren}{g^2} \left(
        1-\cos{g\varphi}-\frac{g^2\varphi^2}{2}
    \right).
\end{equation}
Notice that as $(p^2+m_0^2)^\xi+ \alpha \neq (p^2+m_0^2+\alpha)^\xi$, we are not merely performing a scheme change.
We modify the tropical theory to be closer to the original sine-Gordon model with respect to the mass scaling behaviour of individual Feynman diagrams.

Let $x= g \varphi$. We will expand the solution space of the PDE~\eqref{eq:PDEren} with the boundary 
condition~\eqref{eq:Boundary_Cond_ren_quadratic_subtracted} in powers of $g^2$ while keeping $x$ fixed.
As discussed in Section~\ref{sec:ambiguities}, $\widehat\Gamma^{\text{tr}}_\ren$ is not completely 
determined by these conditions. We will use the following additional renormalization scheme choice,
\begin{align} \label{eq:rencondhat} \widehat\Gamma^{\text{tr}}_\ren (\varphi, \alpha_\ren, g) = -\frac{\alpha_\ren}{g^2} \left( 1-\cos(x)-\frac{x^2}{2} \right) + \text{``higher-order terms''} \,. \end{align}
Here, ``higher-order terms'' include terms of the form $\alpha_\ren \log (\alpha_\ren), \alpha_\ren^2 \log^3 (\alpha_\ren), \ldots$
The reason that such terms appear is that the expansion at $g=0$ is highly singular and resonant.
Our scheme choice here makes sure that, when $\alpha_\ren \rightarrow 1$, the expansion agrees with the diagrammatic expansion where each diagram with a two-valent vertex is set to zero.
We can then expand this modified effective action perturbatively in $g^2$,
\begin{equation} \label{eq:Gammahatseries}
    \widehat{\Gamma}^{\mathrm{tr}}_{\mathrm{ren}}
    = \sum_{L\geq0}g^{2(L-1)}F_L(x,\alpha_\ren) \, , 
    \qquad \text{ with } \qquad
    F_0(x,\alpha_\ren)=-\alpha_\ren\left(1-\cos x -\frac{x^2}{2}\right) \, .
\end{equation}
For example, at one loop, the PDE gives, 
\begin{equation}
\label{eq:quadratic_subtracted_oneloop}
    F_1
    =
    \alpha_\ren
    \left[
        \frac{x^2}{4}\log \alpha_\ren
        -
        \frac{1-\cos(x)}{2}
        \log\!\left(1-\alpha_\ren(1-\cos(x) )\right)
    \right] \, .
\end{equation}
Here, we again find emergent $\log$-terms after the resummation. 
To avoid issues with different branch choices of the logarithm, we assume $ \alpha_\ren > 0 \, . $
This sign is inherited from the non-tropical sine-Gordon theory, which has a positive mass.

If we set $\alpha_\ren=1$, the solution 
$ \widehat\Gamma^{\text{tr}}_\ren $
generates sine-Gordon graphs without two-valent vertices that are weighted by their tropical Feynman measure. 
Because the original sine-Gordon theory does not feature two-valent vertices, we argue that 
$ \widehat\Gamma^{\text{tr}}_\ren $
is a better proxy for the original sine-Gordon theory in the context of studying the perturbative 
expansion in $g$ and its large-order behaviour.

Evaluating $\widehat\Gamma^{\text{tr}}_\ren$ at $\alpha_\ren=1$, for instance using a recursive Ansatz as above, yields the following expansions for the two- and four-point vertices:
\begin{align} \label{eq:quadratic_subtracted_two_four_point}   \left. \partial_\varphi^2 \widehat\Gamma^\tr_\ren \right|_{\varphi=0} &= g^4 +\frac{35}{16}g^6 +\frac{1253}{216}g^8 + \mathcal O(g^{10}) \, , \nonumber\\
  \left. \partial_\varphi^4 \widehat\Gamma^\tr_\ren \right|_{\varphi=0} &= g^2 +3g^4 + \frac{19}{2}g^6 +\frac{1919}{48}g^8 + \mathcal O(g^{10}) \, . \end{align}

We can expand $\widehat\Gamma^\tr_\ren$ in $n$-point functions in the usual way:
\begin{equation}
\widehat\Gamma^\tr_\ren \big|_{\alpha_\ren=1} = \sum_{n, L  \geq 0} \widehat \Gamma_{n,L} \frac{\varphi^n}{n!} {g^{2(L-1)+n}} \, ,
\end{equation}
where the $L$-loop $n$-point function $\widehat \Gamma_{n,L}$ is now a number,
{$\widehat\Gamma_{n,L}=\left.\partial_x^nF_L(x,\alpha_\ren)\right|_{x=0,\,\alpha_\ren=1}$}.

We computed the numbers $\widehat \Gamma_{n,L}$ to $L=40$ loops in exact arithmetic and to $L=60$ loops at $200$-bit precision for $n\in \{2,4\}$. We did so by translating Eqs.~\eqref{eq:PDEren} and \eqref{eq:rencondhat} into a system of recursive ODEs fulfilled by $\partial_x^n F_L(x, \alpha_\ren)|_{x=0}$ and then setting $\alpha_\ren=1$. We find that these data are consistent with the asymptotic form
\begin{equation}
\label{eq:tropical_even_point_asymptotic_ansatz}
    \widehat \Gamma_{2n,L} 
    \underset{L\to\infty}{\sim}
    s^{-L} \,
    \Gamma(L+2n-1)
    \left(
        K_{2n}
        +\frac{a_{2n}}{\sqrt{L}}
        +\frac{b_{2n}}{L}
        +\mathcal O(L^{-\frac{3}{2}})
    \right),
\end{equation}
where the constant $s$ is approximately,
\begin{equation}
\label{eq:tropical_leading_Borel_singularity}
  s \approx 1.518  \, .
\end{equation}
The empirically obtained values of the constants $K_{2n}$, $a_{2n}$, and $b_{2n}$ vary heavily with the precise choice of the fit range in $L$.
So, we do not report fit values for them.
This indicates that the Borel transform of any $n$-point function $ \left. \partial_\varphi^n \widehat\Gamma^\tr_\ren \right|_{\varphi=0} $ will have a singularity at $z \approx 1.518$ in the Borel plane associated with the transformation 
\begin{align} \mathrm B \, \widehat\Gamma^\tr_\ren \big|_{\alpha_\ren=1} = \sum_{n, L \geq 0} \frac{\widehat \Gamma_{n,L}}{L!} \frac{x^n}{n!} {z^{L}} \,. \end{align}

Including the $\sqrt{L}^{-k}$ terms in the Ansatz~\eqref{eq:tropical_even_point_asymptotic_ansatz} improves the fit quality drastically.
Similar non-integer power corrections have been empirically shown to appear in $D=4$ tropical $\phi^4$ theory~\cite{Balduf:2025idm}. The role of such non-integer power corrections to the large-order behaviour is something that is not predicted by older studies of perturbation theory at large-order~\cite{Parisi:1976zh,Collins:1977dw}. An analytic and physical interpretation of these correction terms would be of high interest. We postpone such interpretations to future work.

\subsection{Tropical instantons and tropical QFT at large order}
\label{subsec:tropical-instanton}
Remarkably, the location of the Borel singularity $s$ in~\eqref{eq:tropical_leading_Borel_singularity} can also be obtained through an independent, albeit heuristic, argument. We stress that the discussion in this section is not rigorous and should be understood as a plausibility argument rather than a proof: several steps rely on assumptions that we do not justify in full. A careful treatment remains to be carried out. Nevertheless, the fact that this line of reasoning reproduces the value of $s$ found above suggests that the underlying ideas are sound. Placing them on a rigorous footing would provide a systematic framework for studying tropical QFTs at large order, and we regard this as an interesting problem for future work.

The alternative way of computing $s$ has many similarities to instanton methods in QFT, so we call it a \textit{tropical instanton} analysis. For us, a tropical instanton is a solution of the large-loop-order limit of the tropical loop equation that controls the singularities in the Borel plane. 

In the following, we will make an Ansatz for the asymptotic behaviour of the coefficient functions $F_{L}$ in \eqref{eq:Gammahatseries}. To do so it is convenient to introduce a new notation. As before, we will use $x = g\varphi$, but we will also write $t = \log \alpha_\ren$, which will take the role of a ``time'' variable. Note  that $\alpha_\ren \partial_{\alpha_\ren} = \partial_t$.

At large loop order we expect $F_L$ to take a factorial-over-power form
\begin{equation}
\label{eq:FL}
    F_L(x,\alpha_\ren)
    =
    \Gamma(L+\beta)
    \sum_{\sigma=\pm}\frac{\chi_\sigma(x,t)}{S_\sigma(x,t)^{L}} 
    +  o\left(\frac{\Gamma(L+\beta)}{S_\pm(x,t)^L}\right)
    \, \quad \text{ as } \quad 
    L\to\infty \, ,
\end{equation}
where $\beta$ is a number, whose precise value will not be important to us, and $\chi_\pm$ and $S_\pm$ are functions as indicated. 
The lower-case $o$ means that we neglect all terms that grow slower than $\frac{\Gamma(L+\beta)}{S_+(x,t)^L}$ and $\frac{\Gamma(L+\beta)}{S_-(x,t)^L}$.
We therefore also neglect any possible $1/\sqrt{L}$ corrections that we empirically observed
in \eqref{eq:tropical_even_point_asymptotic_ansatz}. An analytic explanation of these $1/\sqrt{L}$ correction terms would be of high interest, as such correction terms are unexpected in conventional instanton computations. We postpone the interpretation of these correction terms to future~work.

Since $F_L$ is even in $x$ we introduce two branches, $S_-(x,t)=S_+(-x,t)$ and $\chi_-(x,t)=\chi_+(-x,t)$, so that their sum is also even in $x$. Below, $S$ and $\chi$ denote either one of the branches and any equation, including the Hamilton--Jacobi equation~\eqref{eq:HamJac}, holds separately for either $S_{+}$ or $S_{-}$. 
Using this Ansatz, we will aim for expressions for the functions $S_\pm(x,t)$. 

The $L$-loop $2n$-point coefficients from the last section are obtained by differentiating $2n$ times with respect to $x$, setting $x = 0$ and further setting $t=0$ or equivalently $\alpha_\ren=1$, 
\begin{equation} 
    \widehat{\Gamma}_{2n,L} = \left.\partial_{x}^{2n} F_L(x,e^t)\right|_{x=t=0} \, .
\end{equation}
Hence, asymptotically we expect
\begin{equation} \label{eq:Gamma2n}
    \widehat{\Gamma}_{2n,L} 
    \sim
    C_n
    \frac{\Gamma(L+\beta+2n)}{S(0,0)^{L}} \, \quad \text{ as } \quad 
    L\to\infty \, ,
\end{equation}
where $S(0,0) = S_{\pm} (0,0)$, $C_n$ is an appropriate constant, and so the Borel singularity controlling the expansion is unique.  We mention that from Eq.~\eqref{eq:Gamma2n} one can see immediately that in order for the Ans\"atze, \eqref{eq:tropical_even_point_asymptotic_ansatz} and \eqref{eq:FL}, to agree, one should take $\beta=-1$. 

We obtain the equation for $S(x,t)$ by plugging in the Ansatz~\eqref{eq:Gammahatseries}  into the loop equation \eqref{eq:PDEren} and changing coordinates from $\alpha_\ren$ to $t$. Multiplying both its sides by $1-g^2\partial_x^2\widehat{\Gamma}^{\mathrm{tr}}_{\mathrm{ren}}$, we get
\begin{equation}
\label{eq:loopmult}
    \Bigl(1-g^2\partial_x^2\widehat{\Gamma}^{\mathrm{tr}}_{\mathrm{ren}}\Bigr)
    \bigl[-2+(2-g^2)\partial_t\bigr]\widehat{\Gamma}^{\mathrm{tr}}_{\mathrm{ren}}
    =g^2\partial_x^2\widehat{\Gamma}^{\mathrm{tr}}_{\mathrm{ren}} \, .
\end{equation}
We define $M(x,t)$ by
\begin{equation} \label{eq:Mdef}
    M(x,t) = 1 - \partial_x^2F_0 =  1 - e^{t}(1- \cos x) \, ,
\end{equation}
with $F_0$ from~\eqref{eq:Gammahatseries}. Restricted to order $g^{2(L-1)}$ for $L\geq 1$,
Eq.~\eqref{eq:loopmult} reads in this notation
\begin{equation}
\label{eq:trloopas}
    M \Big((-2+2\partial_t)F_L-\partial_tF_{L-1} \Big)
    =\partial_x^2F_{L-1}
    +\sum_{A=1}^{L-1}\partial_x^2F_A \, \Big((-2+2\partial_t)F_{L-A}-\partial_tF_{L-A-1} \Big) \, .
\end{equation}
In particular, at one loop we recover \eqref{eq:quadratic_subtracted_oneloop},
\begin{equation}
\label{eq:F1}
    F_1(x,t)
    =e^t\left[
        \frac{x^2}{4}t
        -\frac{1-\cos x}{2}
        \log M
    \right] \, .
\end{equation}
Using $(L+\beta)\Gamma(L+\beta) = \Gamma(L+\beta+1)$ and keeping only the dominant growth rate gives,
\begin{equation}
    \partial_tF_L\sim- (\partial_tS)\;\chi\,
    \frac{\Gamma(L+\beta+1)}{S^{L+1}} \, ,
    \qquad
    \partial_x^2F_{L-1}\sim(\partial_xS)^2\,\chi\,
    \frac{\Gamma(L+\beta+1)}{S^{L+1}} \, .
\end{equation}
So, in Eq.~\eqref{eq:trloopas} the leading terms are $2M\partial_tF_L$
on the left, and $\partial_x^2F_{L-1}$ together with the boundary
term $A=L-1$ of the sum on the right.  One can see that the other terms are subleading. Cancelling the common factor
$\chi\,\Gamma(L+\beta+1)/S^{L+1}$ and using \eqref{eq:F1} gives
\begin{equation}
\label{eq:HamJac}
    -2M\,\partial_tS
    =(\partial_xS)^2\,\Biggl(1+\frac{1-M}{M}\Biggr)
    =\frac{(\partial_xS)^2}{M}
    \Rightarrow
    \partial_t S=-\frac{(\partial_x S)^2}{2M^2} \, .
\end{equation}
This is the textbook \textit{Hamilton--Jacobi equation} for a non-relativistic particle of position- and time-dependent ``mass'' $M(x,t)^2$ (note the square) moving in one dimension. A more direct route to \eqref{eq:HamJac} is given by inserting the instanton Ansatz
\begin{equation}
\label{eq:transseries}
    \widehat\Gamma^{\tr}_{\ren}=\Gamma_{0} + e^{-S(x,t)/g^2} \Gamma_{1} + \ldots  \, ,
\end{equation}
and treating the $\exp( {-S(x,t)/g^2})$ as a new formal parameter in the PDE. %
Here, $\Gamma_1$ is the first instanton fluctuation of the effective action.
See, e.g., \cite{Borinsky:2016mqu,Borinsky:2017hkb} for an explanation for the equivalence of this Ansatz with the large-loop order analysis.
The physical perspective on the resurgence framework provides more concrete interpretations of the correspondence between instanton and large-order perturbative phenomena (see, e.g., \cite{dunne2026introductory}). We believe that a resurgence analysis of tropical sine-Gordon theory and the tropical loop equation would be a highly fruitful endeavor. However, such an analysis is beyond the scope of the present paper.

We hence see that the large loop-order behaviour of the perturbation series is governed by classical mechanics with the field considered as a coordinate of a particle moving in ``time'' $t$, which encodes the coupling $\alpha_\ren$. The location of the Borel singularity $S(0,0)$ is thus given by the value at the origin of the classical action, considered as a function of the endpoint of the particle's trajectory. 

Recall that the general solution to \eqref{eq:HamJac} is given by the familiar classical action
\begin{equation} \label{eq:ClasAction}
    S = S(x,t) = \int_{t_0}^{t} \dd \tau \mathcal{L}(x(\tau),\dot{x}(\tau),\tau) \, ,\quad \text{ where } \quad \mathcal{L}(x,\dot{x},t) = \frac{M(x,t)^2}{2} \dot{x}^2 \, ,
\end{equation}
evaluated on the extremizing trajectory given by the solution to the Euler--Lagrange equation,
\begin{equation} 
\label{eq:ELeq}
    M \ddot x=-  \dot x^2 \, \partial_x M - 2  \dot x \, \partial_t M\, ,
\end{equation}
with initial conditions at $t=t_0$ for the position and velocity of the ``particle''. 
As usual, $\dot x$ is the $t$-derivative of $x$.
Given an initial velocity at the point $(x_0,t_0)$ we can  solve the Euler--Lagrange equation and compute the action at the point $(x,t) = (0,0)$ by plugging in the trajectory connecting the two points into the integral \eqref{eq:ClasAction}. 

Obtaining such initial conditions is challenging. Here, we propose the following construction.
We define a surface $\Sigma$ in $(x,t)$ space:
\begin{equation}
    \Sigma:=\{(x,t):M(x,t)=0\} \, \, .
\end{equation}
Note that the coefficient functions $F_{L}(x,t)$ diverge when $(x,t)$ is close to the surface $\Sigma$.
For example, $F_1(x,t)$ given in Eq.~\eqref{eq:F1} has a logarithmic singularity on $\Sigma$. Generally, from \eqref{eq:trloopas} we can infer the behaviour
\begin{equation}
\label{eq:FMScaling}
    F_L \sim \frac{1}{M^{3(L-1)}} \, \, \text{for } \, M \rightarrow 0 \, .
\end{equation}
Indeed, for $L=1$ near $M=0$ we have from \eqref{eq:F1} the scaling $F_1 \sim \log{M}$. Then, assuming the behaviour \eqref{eq:FMScaling}, and using $\partial_t = - \partial_M$ near $M=0$, we can write\footnote{This is only valid near a generic point of $\Sigma$, $M_x\neq0$.}
\begin{equation}
 \partial_tF_k\sim M^{2-3k},
 \qquad \partial_x^2F_k\sim M^{1-3k} \, ,
 \qquad 1\leq k<L \, .
\end{equation}
Every term in the convolution sum on the RHS of \eqref{eq:trloopas} then gives at leading order
\begin{equation}
 \partial_x^2F_A\,2\partial_tF_{L-A}
 \sim M^{1-3A}M^{2-3(L-A)}=M^{3-3L} \, .
\end{equation}
The separate term $\partial_x^2F_{L-1}$ is one power less singular. Thus we get
\begin{equation}
 2M\partial_tF_L\sim M^{3-3L}
 \Rightarrow F_L\sim M^{-3(L-1)} \, .
\end{equation}

On the real $M>0$ side, the leading coefficients are positive by
induction from $F_1$, so they do not cancel. All convolution terms
contribute at the leading pole order. The exceptional points
$M_x=0$ must be treated separately. The $M$-dependence of $F_L$ in \eqref{eq:FMScaling} together with the general
Ansatz \eqref{eq:FL} suggests that (at least locally) the zeroes of $S(x,t)$ lie on the surface $\Sigma$. We now make the assumption that the relevant trajectory starts off from the surface $\Sigma$, where the action $S$ vanishes. For trajectories starting from a generic point of $\Sigma$ (i.e.\ one with $M_x\neq0$), the Hamilton--Jacobi
equation \eqref{eq:HamJac} then fixes the order of vanishing of $S$ near $\Sigma$,
\begin{equation} \label{eq:SMScaling}
    S(x,t) \sim M(x,t)^3 \,\, ,
\end{equation}
since an Ansatz $S\propto M^p$ solves \eqref{eq:HamJac} near $\Sigma$ with a finite non-zero coefficient only for $p=3$. This is consistent with \eqref{eq:FMScaling}: the factor $M^{-3(L-1)}$ of the leading singularity at fixed $L$ matches $S^{-L}$ in \eqref{eq:FL}. With this assumption fixed, we obtain below only two valid trajectories, one of which reproduces the singularity location found in the perturbative analysis of the previous sections.

A simple argument using Eq. \eqref{eq:HamJac} shows how the condition \eqref{eq:SMScaling} translates into the initial condition for the particle: given some $x_0$ on $\Sigma$ the initial velocity is 
\begin{equation}
    \dot{x} \bigr|_{\Sigma} = -2 \tan{\frac{x_0}{2}} \, .
\end{equation}
The problem thus boils down to finding the right value for the pair $(x_0,t_0)$ on the surface $\Sigma$, such that the particle lands at the origin $x=0$ at time $t=0$. We limit attention to the region $x \in (0, 2\pi)$ which is one period interval for the original perturbative solution. For convenience we introduce an auxiliary variable $r$ parameterizing a given trajectory through $(x,t)$-space,
\begin{equation}
    r = 2 \sqrt{-e^t M \dot{x}} \, ,
\end{equation}
where $M \dot{x}$ is always negative on the trajectories of interest. Indeed, using~\eqref{eq:ELeq}, one has $d (M \dot{x})/dt=-M_t\dot x$, hence $d(M \dot{x})/dx=-M_t=e^t(1-\cos x)\ge0$. For trajectories starting from $M=0$ and traveling to the left with decreasing $x$ (see Figure~\ref{fig:trajinst}), the variable $r$ starts from zero and increases. In terms of $r$, the Euler--Lagrange equation \eqref{eq:ELeq} can be rewritten as the system
\begin{align} \label{eq:ELsystem} \frac{dx}{dr} &= - \frac{r}{2} e^{-t} \, ,\\
 \frac{dt}{dr} &= \frac{2 M}{r} \, . \end{align}
One observes that $x$ decreases along any extremizing trajectory, whereas the flow of ``time'' is generally not monotonic (this is, of course, not a problem since $t$ is not an actual physical time).  The action accumulated along an extremizing trajectory is then given by 
\begin{equation}
    \frac{dS}{dr} = \frac{r}{4} \left(\left(\frac{dx}{dr}\right)^2 - x + \sin{x} \right) +\frac{d}{dr}\left[\frac{r^2}{8}(x-\sin x)\right]\, ,
\end{equation}
where the boundary term vanishes at the launch $r=0$ and at the endpoint $x=0$. We integrate the system in $r$ until we reach the value $r=r^*$ such that $x(r^*) = 0$ and record the corresponding value of time, using the relation $e^{-t} = - \frac{2}{r} \frac{dx}{dr}$. We then scan over all $x_0$ to get the landing at $t=0$. There are, in fact, only two real trajectories, displayed in green (dashed) and red (solid) in Figure~\ref{fig:trajinst}, satisfying these requirements.  
The red one gives precisely the value observed in the numerical analysis above,
\begin{equation}
\label{eq:instanton_S00}
    S(0,0) \simeq 1.518 \, .
\end{equation}
Importantly the green trajectory has a \textit{smaller} absolute value of $S(0,0)$, and so, with all things being equal, should dominate the asymptotic expansion. Supported by the data, we hypothesize that there must exist a selection mechanism that rules out this trajectory and leaves only the red one. The examination of that question would require a more thorough analysis with the derivation of the Stokes constants $\chi$ corresponding to both trajectories.

Curiously, there also exists a trajectory starting off from the ``minimum of the well'' (the lowest point of the surface $\Sigma$) that also arrives at $(x,t)=(0,0)$. It does not, however, satisfy the scaling \eqref{eq:SMScaling} and in that sense is disconnected from the perturbation theory.

We finally note also that the method is perfectly suitable for analyzing the large-order behaviour of the $\alpha_\ren$-dependent $\Gamma^{\tr}_\ren$. It satisfies the same tropical loop equation, but with a different boundary condition which will yield a different surface $\Sigma$. 
\begin{figure}[t]
\centering
\includegraphics[width=0.94\linewidth]{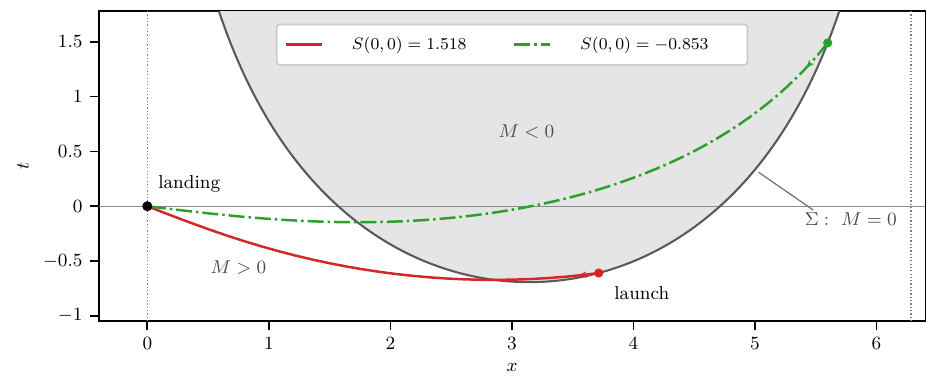}
\caption{The surface $\Sigma$: $t=-\log(1-\cos x)$ for $x \in (0 , 2 \pi)$, and the two trajectories (green/dashed and red/solid) landing at $(x,t)=(0,0)$. The red one has $S(0,0) \simeq 1.518$. If the action branch associated with the red trajectory is denoted by $S_+$, its reflection $S_{-}(x,t) = S_{+}(-x,t)$ is generated by the mirrored trajectory with $x \rightarrow -x$.}
\label{fig:trajinst}
\end{figure}
\section{Discussion and outlook}
In this article, we investigated the nature of tropical QFT.
After reviewing tropical QFT in general dimensions, we focused on tropicalizations of two-dimensional QFTs
and classified the space of scale-invariant tropical QFTs in that dimension that have a periodic effective action.
One of the main results of this paper is that each of these periodic scale-invariant tropical QFTs can be reached by flowing a sine-Gordon Lagrangian to its IR fixed point. This motivated our detailed study of the two-dimensional tropical sine-Gordon model. 
Mirroring Coleman's seminal work on the sine-Gordon model, we performed the necessary renormalization of the tropical QFT and constructed a convergent expansion in the mass-coupling $\alpha_\ren$ at fixed non-resonant $0<g^2<2$. Building on this, we  numerically solved the PDE, the renormalized tropical loop equation, which completely describes the tropical sine-Gordon model.  This numerical solution allowed us to visualize the full renormalization group flow of the tropical sine-Gordon model from the Gaussian fixed point to its IR fixed point.
Further, we illustrated that information from the local potential approximation from the functional renormalization group framework can be used beneficially to study tropical QFT.

We also investigated the large-loop-order behaviour of the perturbative expansion of a slightly modified version of tropical sine-Gordon QFT which is more amenable to a large-order analysis.
The growth rate of the coefficients is  consistent with factorial growth and exhibits a Borel singularity at $s\simeq1.518$. 
We outlined how to obtain these results via an independent computation: in the large loop-order limit, the tropical loop equation takes the form of a classical Hamilton--Jacobi equation. The solutions, which we call tropical instantons, reproduce the value of $s$ within the accuracy of the perturbative data. 

Our discussion of the large-order aspects and the relationship to classical mechanics remains open-ended.
One natural next step would be to extract other asymptotic information, such as the power-law and constant multiplicative behaviour of the coefficients, and match it to the empirical data from the perturbative solution. This would confirm the approach's validity and help determine which trajectories contribute. In particular, the role of the second regular trajectory remains open. Understanding the associated Stokes data would be necessary to treat the Borel ambiguity. The same idea applies to the bare tropical theory and to other potentials obeying the tropical loop equation.

\subsection*{Acknowledgments}
We thank Paul-Hermann Balduf, Dario Benedetti, David Berenstein, Roland Bittleston, Erik Panzer and Vladimir Kazakov for insightful discussions. The work of JR and SRK is supported by the Natural Sciences and Engineering Research Council (NSERC) of Canada. Research at Perimeter Institute is supported in part by the Government of Canada through the Department of Innovation, Science and Economic Development Canada and by the Province of Ontario through the Ministry of Colleges and Universities. The work of JR is also supported by the Villum Young Investigator Programme under project No. 71589. The Center of Gravity is a Center of Excellence funded by the Danish National Research Foundation under grant No. 184.

\subsection*{AI disclosure statement}
The authors used Claude Fable 5.0, 5.1, Opus 5 and GPT-6 Astra to write computer programs for empirical investigations of the generating functions, generating the displayed graphics and the numerical integration of the tropical loop equation using the Ansatz~\eqref{eq:ansatz_numerical}. Fable 5.0 was used for the derivation of the instanton equation ~\eqref{eq:HamJac}, its subsequent simplification and writing the code for the numerical integration of~\eqref{eq:ELsystem}. Further, we used these AI tools for literature research, to check our arguments and the general write-up for typos and errors. Of course, the authors take full responsibility for all claims and arguments in the paper.

\printbibliography

\end{document}